\documentclass[10pt,journal,epsfig,twoside]{IEEEtran}
\usepackage{amsmath}
\usepackage{graphicx} 
\usepackage{epstopdf}
\usepackage{amssymb}
\usepackage{amsfonts}
\usepackage{amsthm}
\usepackage{cite}
\usepackage{bm,comment}
\usepackage{algorithm}
\usepackage{algpseudocode}

\usepackage{subeqnarray}
\usepackage{subfigure}
\usepackage{multicol}
\usepackage{multirow}
\usepackage{diagbox}
\usepackage{slashbox}
\usepackage{stfloats}
\usepackage{float}
\usepackage{color} 
\usepackage{cases}
\usepackage{setspace}

\floatname{algorithm}{Algorithm}
\renewcommand{\algorithmicrequire}{\textbf{Given:}} 
\renewcommand{\algorithmicensure}{\textbf{Return:}}

\newtheorem{proposition}{Proposition}
\newtheorem{remark}{Remark}

\renewcommand{\baselinestretch}{0.97}
\graphicspath{{figs/}}

\begin{document}
\IEEEoverridecommandlockouts

\title{Broadcasting in Centimeter-Wave Gigantic MIMO Systems: Beam Shaping is All You Need}
\author{Boyu Ning,  Yike Xie, Jaebum Park, Songjie Yang, Weidong Mei, Emil Bj{\"o}rnson, and Robert W. Heath Jr.

\thanks{Boyu Ning and Emil Bj{\"o}rnson are with Department of Communication Systems, KTH Royal Institute of Technology, 100 44 Stockholm, Sweden (email: boydning@outlook.com; emilbjo@kth.se). This work was supported by the Grant 2022-04222 from the Swedish Research Council.}     

\thanks{Yike Xie and Weidong Mei are with National Key Laboratory of Wireless Communications, University of Electronic Science and Technology of China, Chengdu 611731, China (e-mails:  202411220601@std.uestc.edu.cn; wmei@uestc.edu.cn).}

\thanks{Songjie Yang is with the Beijing National Research Center for Information Science and Technology (BNRist), Tsinghua University, Beijing 100084, China. (e-mails:  18101979553@163.com).}

\thanks{Jaebum Park and Robert W. Heath Jr. are with the Department of Electrical and Computer Engineering, University of California at San Diego, San Diego, CA 92093 USA (e-mail: jp261@ucsd.edu; rwheathjr@ucsd.edu).}

}

\maketitle

\begin{abstract}
Beamforming is a key enabling technology for centimeter-wave Gigantic MIMO (G-MIMO) systems, using antenna arrays to generate wide beams for cell-wide broadcasting or narrow beams for data transmission. This paper presents a novel \emph{beam shaping technology} for designing environment-aware broadcast beams (i.e., beams for common control information delivery, as in 5G SSB). As a case study, we shape a \emph{prostrate beam} in the vertical plane to approximate uniform electromagnetic power distribution across the ground, thereby enhancing broadcast coverage quality. We first propose a geometric broadcasting model that, through coordinate transformation, relates base station height, beam spread, maximum service distance, and downtilt angle. Building on this model and the signal propagation model, we formulate the beam shaping problem as a non-convex optimization problem that aims to maximize ground-level coverage. To solve it, we develop a generalized optimization framework. We reformulate the problem into a matrix representation, equivalently transforming it into a target pursuit optimization (TPO) problem. Using variable substitution, the objective is convexified into a quadratically constrained quadratic program, and the non-convex rank constraint is handled via successive convex approximation, decomposing the problem into iteratively solved subproblems. We provide complexity and convergence analyses for the proposed algorithm. Numerical simulations confirm that the optimization effectively shapes the intended beam pattern and that the proposed prostrate beam yields significant performance gains, including substantially extended coverage distance, over conventional wide-beam benchmarks.
\end{abstract}

\section{Introduction}
 Multiple-input multiple-output (MIMO) technology has fundamentally transformed wireless communications over the past three decades\cite{emilten, HeathLozanoFoundationsMIMOCommunications2018}. By employing multiple transmit and receive antennas, MIMO enables the simultaneous transmission of multiple data streams, thereby dramatically increasing spectral efficiency. In the Universal Mobile Telecommunications System (UMTS) 3G era, base stations typically employed a simple $1$T$2$R (one transmit, two receive) antenna configuration, primarily leveraging receive diversity. During the Long Term Evolution (LTE) 4G era, the most common base station configuration featured $2$T$2$R or $4$T$4$R digital antenna ports, which supported basic MIMO diversity and multiplexing functions but lacked the capability for precise beam direction control\cite{4t4r}. In the New Radio (NR) 5G era, massive MIMO configurations typically incorporate $16$, $32$, or $64$ digital antenna ports, enabling stronger directional beamforming and delivering a potential peak rate of approximately $20\times$ compared to 4G systems\cite{ning2026precoding, 16tr, DreifuerstHeathMassiveMIMO5GHow2023}.

 In addition to the increase in the number of antennas, the evolution of wireless communications has involved adopting higher carrier frequencies to access more available bandwidth. Consequently, 5G was co-designed for two frequency ranges: FR1 ($0.4–7$ GHz) and FR2 ($24–71$ GHz)\cite{dilli2020analysis}. FR2, also known as the millimeter-wave band, provides greater bandwidth to support high data rates and facilitates the integration of a large number of antennas into compact arrays\cite{HeathSIP}, thereby enabling key massive MIMO techniques such as hybrid analog-digital beamforming \cite{HeathOMP}, beam training\cite{training, AlkhateebEtAlInitialBeamAssociation2017, DreifuerstHeathHierarchicalMLCodebookdesign2024}, and beam tracking\cite{tracking}. However, due to the propagation characteristics of high frequencies, mmWave systems face persistent challenges such as signal blockage and poor outdoor-to-indoor penetration\cite{Mohebi}. As a result, mmWave technology has not achieved widespread commercial success or large-scale deployment. Similarly, communication systems in even higher frequency bands, such as sub-terahertz, also confront significant propagation-related obstacles\cite{ning2023beamforming}.

In 2025, 3GPP initiated standardization discussions for 6G\cite{lin2025bridge}. There is growing interest from both academia and industry in the underutilized centimeter-wave bands spanning 7–24 GHz, also referred to as the upper mid-band or FR3\cite{tang2025preliminary}. These bands offer a compelling trade-off: they provide substantially more bandwidth than traditional sub-6 GHz frequencies, while exhibiting better coverage and penetration characteristics than mmWave bands\cite{testolina2024sharing,kang2024cellular,cui20256g}. This frequency range enables MIMO systems to support base stations with at least 256 antenna ports, or even thousands of antenna ports in advanced configurations, marking the advent of a new era often termed Gigantic MIMO (G-MIMO)\cite{Gigantic}. Depending on the array scale or number of antennas, different nomenclatures may emerge, e.g., extremely large‑scale MIMO (XL-MIMO) for general large arrays\cite{wang2024tutorial}, and ultra‑massive MIMO (UM‑MIMO) for even larger numbers of antennas, which is the prevailing terminology in terahertz systems\cite{ning2023beamforming}. This evolution brings forth potential new technologies such as tri-hybrid architectures\cite{heath2025tri,castellanos2025embracing}, near-field beam focusing\cite{XL1,XL2, HeathGonzalez-PrelcicBeamsharingMixedNearfieldfarfield2024, nearpeng}, and dynamic antenna movement\cite{ning2025movable, zhu2025tutorial, 6dma, weidong, songjie, CarlsonEtAlWidebandDynamicMetasurface2025a}.

Broadcasting is a critical functionality of cellular systems, as it enables initial user access and fundamental cell management. This broadcast coverage must balance wide reach within the cell with minimal interference to its neighbors. Consequently, array downtilt has been a common technique since the UMTS 3G era to improve cell isolation and/or increase coverage\cite{umtstilt}. Array downtilt can be adjusted mechanically\cite{mechanical} and/or electrically\cite{electrical}. Various techniques for electrical tilt exist, such as remote electrical tilt (RET) and variable electrical tilt (VET)\cite{ltetilt}\footnote{ The use of electrical downtilt eliminates the need for base station site visits and tower climb by controlling the tilt angle via the network management system, so that operational cost is saved; however, mechanical down tilt is also needed because the electrical tilt range is limited.}. In the LTE 4G era, with a frequency reuse factor of one, array downtilt is likely to be even more important for achieving good cell isolation and, hence, high system performance\cite{tiltcomp}. The coverage beamwidth of LTE broadcast signals is mainly determined by the antenna radiation pattern, and the antenna tilting angle is optimized to achieve the best coverage\cite{3gpptilt}. With the advent of 5G and beyond, it is essential to optimize not only antenna downtilt but also the broadcast beam design. The term ``broadcast'' refers to the transmission of cell‑wide common information to all users within a cell sector. Specifically, it is realized through the synchronization signal block (SSB), which contains the physical broadcast channel (PBCH) carrying the master information block (MIB), as well as system information blocks (SIBs). The broadcast signal, i.e., SSB, is transmitted via multiple sector‑wide beams to cover the entire cell sector\cite{DreifuerstH022,bang2021millimeter}. These wide beams are formed by multiple antenna elements working jointly to cover a common vertical range while addressing independent azimuthal sectors\cite{ssbhang}. 

The synthesis of wide beams can effectively utilize windowing techniques from digital filter design\cite{wu2005study}. Given a fixed number of antenna elements, one can apply classical window functions such as Hamming, Hann, or Blackman to shape beams of different widths\cite{gopika2015fir}. However, the beamwidth attainable with these windows is primarily dictated by the window type itself, and fine adjustment would typically require changing the antenna count (equivalent to the filter order). To achieve flexible beamwidth tuning without altering the array size, several heuristic closed-form solutions have been developed. These include methods based on the superposition of pencil beams\cite{li2017shaped}, sharpened raised-cosine finite impulse response (FIR) filters \cite{molnar2018sharpened}, piecewise polynomial beam shaping\cite{toso2022simple}, modified Dolph-Chebyshev polynomial \cite{molnar2024synthesis}, polynomial transforms of Gaussian weights \cite{molnar2023synthesis}, and sum of symmetrical array response vectors (S-SARVs) \cite{trad4}. The smoothness of a wide beam can be further improved by applying iterative optimizations\cite{Broadbeam,trad3,ma2022passive}. To realize feasible phased-array implementations, several workarounds have been investigated from different perspectives, including hybrid beamforming architectures \cite{hb1,hb2,zhang2019power}, multi-subarray designs \cite{zhu20193, ning2024codebook}, and dual-polarization techniques \cite{girnyk2021efficient, petersson2022energy}.

Given that centimeter-wave G-MIMO systems can support thousands of digital ports, or even more in hybrid antenna architectures, a key challenge lies in balancing complexity and performance for broadcasting\cite{massivemimobook}. Specifically, coverage capability is improved by increasing the number of broadcast beams: more beams result in narrower beamwidths and higher beamforming gain, but also increase the time required for beam sweeping during the initial access process\cite{giordani2018tutorial}.  Even after optimizing the downtilt angle and the number of broadcast beams, coverage imbalance remains due to non‑uniform propagation conditions. Conventional broadcast beam design assumes uniform gain across the main lobe, but in practice, path loss and blockage cause large variations in received power across users. To address this, we propose a redistributive \emph{beam shaping} approach that reduces beamforming gain in directions with excessive received power and increases it in directions with insufficient power, i.e., an effect we refer to as the \emph{Robin Hood principle} in beam design. This effect refers to a resource allocation strategy that ``takes from the rich and gives to the poor,'' which in our context means that the beamforming gain is deliberately reduced in strong-signal directions and redistributed to weak-signal or coverage-limited directions. 

 Motivated by this, in this paper, we propose a beam shaping technique tailored for broadcasting in centimeter-wave G-MIMO communication systems\footnote{For sub-6 GHz systems, the limited number of antennas makes beam shaping difficult, and uniform wide beams are typically used for broadcasting. For mmWave systems, high-gain narrow beams are required to overcome severe path loss. In contrast, centimeter-wave systems offer a sufficient number of antennas while experiencing less severe propagation loss than mmWave, which motivates the beam shaping approach proposed in this work. }. Conventional beamforming focuses on enhancing signal strength or suppressing interference in \emph{target directions}. In contrast, beam shaping aims to engineer the \emph{entire beam pattern} to adapt to the physical environment, assuming the environmental prior information is available\cite{radioprofile}. Specifically, the broadcasting beam pattern can be optimized following the Robin Hood principle. For example, in the vertical plane, the beam pattern can be tailored to the ground geometry and its associated propagation conditions to balance the received power. To this end, we propose a novel prostrate beam design as a realization of beam shaping in the vertical plane.\footnote{Beam shaping can also be utilized in the horizontal plane by leveraging environmental priors \cite{radioprofile}, whereby more radiation power is directed toward buildings to improve outdoor-to-indoor (O2I) performance, while less power is allocated to line-of-sight (LoS) areas or terminal-free zones.} Unlike conventional wide beams that have uniform radiation power across the coverage directions, the prostrate beam ensures uniform power distribution along the ground, thereby enhancing broadcast performance. The main contributions of this paper are summarized as follows:

\begin{itemize} 
\item We propose a novel geometric broadcasting model. Using coordinate transformation techniques, we analyze the beam geometry of a tilted antenna array in a global coordinate system, establishing relationships among the base station height, beam spread (i.e., coverage width), maximum service distance, and downtilt angle.

\item Based on the geometric broadcasting model and signal propagation model, we formulate a beam shaping problem to generate a prostrate beam. The objective is to design a beam whose radiated power appears to ``crawl" along the ground, maximizing ground-level coverage. However, the resulting beam shaping problem is generally non-convex and difficult to solve directly.

\item To address this class of beam shaping problems, we propose a general optimization framework. We first reformulate the problem into a matrix representation, equivalently transforming it into a target pursuit optimization (TPO) problem. Using variable substitution techniques, the objective function is convexified into a quadratically constrained quadratic program (QCQP). To handle the non-convex rank constraint, we employ the successive convex approximation technique, decomposing the problem into several subproblems that are solved iteratively.

\item We provide complexity and convergence analyses for the proposed algorithm. Simulation results demonstrate that our optimization method effectively shapes the beam pattern as intended. Compared to traditional wide beams, the proposed prostrate beam achieves significantly extended coverage distance.

\end{itemize}

The remainder of this paper is organized as follows. Section II presents the system model. Section III provides the motivation and formulates the beam-shaping problem for prostrate beam design. Section IV presents our proposed methods. Section V evaluates the performance via numerical results. Finally, Section VI concludes the paper.

\begingroup
\allowdisplaybreaks
\section{System Model}
We consider a downlink G-MIMO communication system, where the BS is equipped with a large antenna array for broadcasting a public message to system users. 
\subsection{Array Model}
We assume the BS is equipped with a uniform planar array (UPA) with $N_a=MN$, which can simultaneously control horizontal and vertical radiation, achieving three-dimensional (3D) beam control. The coordinate system and the orientation of the array are depicted in Fig. \ref{array}, where $M$ is the number of horizontal antenna elements and $N$ is the number of vertical antenna elements. $d_x$ and $d_z$ denote the horizontal and vertical inter-antenna spacing, respectively. 

 \begin{figure}[t]
\centering
\includegraphics[width=3.4in]{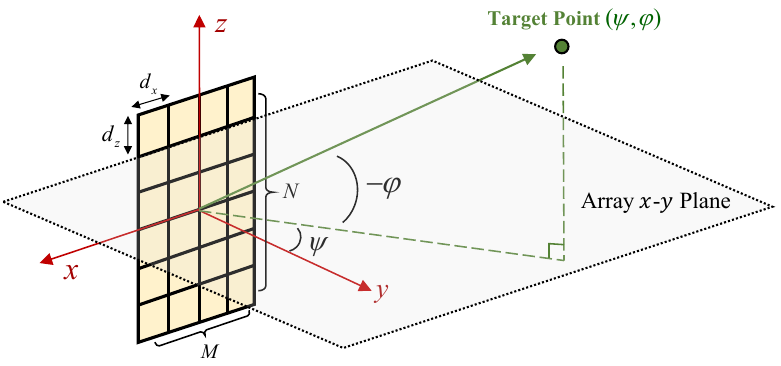}
\caption{Illustration of the UPA geometry. The array has $M$ horizontal and $N$ vertical elements with half‑wavelength spacing $d_x = d_z = \lambda/2$. A far‑field target point is characterized by azimuth $\psi$ and elevation $\varphi$.}\label{array}
\vspace{-8pt}
\end{figure}

Following the conventional UPA model, we consider a half-wavelength spacing $d_z=d_x=\lambda/2$ and count from $0$ the row and column indices. Then, for a far-field target point at azimuth \(\psi \in [-\pi/2,\pi/2]\) and elevation \(\varphi \in [-\pi/2,\pi/2]\), the array response matrix \(\mathbf{A}(\psi, \varphi) \in \mathbb{C}^{N \times M}\) captures the phase progression across a uniform rectangular array with \(N\) vertical and \(M\) horizontal elements. Its \((n,m)\)-th entry is given by  
\begin{equation}
\left[ \mathbf{A}(\psi, \varphi) \right]_{n,m} = \frac{1}{\sqrt{NM}} e^{j\pi (m \cos\varphi \sin\psi + n \sin\varphi )},
\end{equation}
where $n=0,1,...,N-1$ and $m=0,1,...,M-1$. The exponential term represents the relative phase shift of the antenna element located at the \(n\)-th vertical and \(m\)-th horizontal position relative to the reference element at \((0,0)\). The factor \(\frac{1}{\sqrt{NM}}\) ensures power normalization, while the expression \(m \cos\varphi \sin\psi\) and \(n \sin\varphi\) account for the phase differences along the horizontal and vertical directions due to the incoming planar wave. 

Let ${\bf{W}} \in {\mathbb{C}}^{N\times M}$ denote the beamforming matrix. The beam radiation in the direction of $(\psi ,\varphi)$ is given by
\begin{equation}\label{radr}
r(\psi ,\varphi ) = \sum\nolimits_{n,m} {{{\left[ {{\bf{A}}^*(\psi ,\varphi ) \circ {\bf{W}}} \right]}_{n,m}}}
\end{equation}
where $ ^*$ is the conjugate operator and  $\circ$ is the Hadamard product operator. We assume that the beamforming weights in the horizontal and vertical directions are designed independently, and let \(\mathbf{w}_h\) and \(\mathbf{w}_v\) denote the horizontal and vertical beamforming vectors, respectively. Then, the beamforming matrix is subject to 
\begin{equation}\label{ww}
    {\bf{W}}={\bf{w}}_v \otimes{\bf{w}}_h^T,
\end{equation}
where $\otimes$ is the Kronecker product operator. The array response vectors in vertical and horizontal directions can be respectively expressed as ${{\bf{a}}_v}(\psi ,\varphi )$ and ${{\bf{a}}_h}(\psi ,\varphi )$, i.e.,
\[
{{\bf{a}}_v}(\psi ,\varphi ) \!=\! {{\bf{a}}_v}(\varphi ) \!=\! \frac{1}{{\sqrt N }}\left[ {1,{e^{j\pi \sin \varphi }},...,{e^{j\pi (N - 1)\sin \varphi }}} \right]^T
\]
and
\[
{{\bf{a}}_h}(\psi ,\varphi ) \!=\! \frac{1}{{\sqrt M }}\left[1,{e^{j\pi \cos \varphi \sin \psi }},...,{e^{j\pi (M - 1)\cos \varphi \sin \psi }}\right]^T.   
\]
Then, the beam radiation to the target point in the direction of $(\psi ,\varphi)$ can be decoupled by
\begin{align}\label{ra}
r(\psi ,\varphi ) &= \sum\limits_{n,m} {{{\big[ {{{\bf{A}}^*}(\psi ,\varphi ) \circ {\bf{W}}} \big]}_{n,m}}} \notag \\
 &= \sum\limits_{n,m} {{{\left\{ {\left[ {{\bf{a}}_v^*(\psi ,\varphi ) \otimes {{\bf{a}}_h}{{(\psi ,\varphi )}^H}} \right] \circ \left[ {{{\bf{w}}_v} \otimes {\bf{w}}_h^T} \right]} \right\}}_{n,m}}} \notag\\
&\mathop  = \limits^{(a)} \sum\limits_{n,m} {{{\left\{ {\left[ {{\bf{a}}_v^*(\psi ,\varphi ) \circ {{\bf{w}}_v}} \right] \otimes \left[ {{{\bf{a}}_h}{{(\psi ,\varphi )}^H} \circ {\bf{w}}_h^T} \right]} \right\}}_{n,m}}} \notag\\
&\mathop  = \limits^{(b)} \left[ {\sum {{\bf{a}}_v^*(\psi ,\varphi ) \circ {{\bf{w}}_v}} } \right] \cdot \left[ {\sum {{\bf{a}}_h^*(\psi ,\varphi ) \circ {{\bf{w}}_h}} } \right] \notag\\
& = \left[ {{{\bf{a}}_v}{{(\psi ,\varphi )}^H}{{\bf{w}}_v}} \right] \cdot \left[ {{{\bf{a}}_h}{{(\psi ,\varphi )}^H}{{\bf{w}}_h}} \right],
\end{align}
where equality $(a)$ follows from the identity 
\begin{equation}
(\mathbf{A} \otimes \mathbf{B}) \circ (\mathbf{C} \otimes \mathbf{D}) = (\mathbf{A} \circ \mathbf{C}) \otimes (\mathbf{B} \circ \mathbf{D})    
\end{equation}
and equality $(b)$ holds due to the fact that
\begin{equation}
\sum_{n=1}^{N}\sum_{m=1}^{M} [\mathbf{a} \otimes \mathbf{b}^T]_{n,m} = \sum_{n=1}^{N}\sum_{m=1}^{M} a_n b_m= \sum_{n=1}^{N}a_n\sum_{m=1}^{M} b_m.  
\end{equation}
Equation (\ref{ra}) implies that the 3D beam radiation power is the product of two-dimensional (2D) beam radiation powers, i.e., 
\begin{equation}\label{dpower}
 \begin{aligned}
p(\psi ,\varphi ) &= {\left| {r(\psi ,\varphi )} \right|^2}\\
 &= {\left| {{{\bf{a}}_v}{{(\psi ,\varphi )}^H}{{\bf{w}}_v}} \right|^2} \cdot {\left| {{{\bf{a}}_h}{{(\psi ,\varphi )}^H}{{\bf{w}}_h}} \right|^2}.
\end{aligned}   
\end{equation}
Consequently, one can design the beamforming vectors based on 2D patterns in horizontal and vertical, respectively, and obtain the 3D beamforming vector by using a Kronecker product on the 2D beamforming vectors, as presented in (\ref{ww}).
\subsection{Geometric Broadcasting Model}
\begin{figure}[t]
\centering
\includegraphics[width=3.4in]{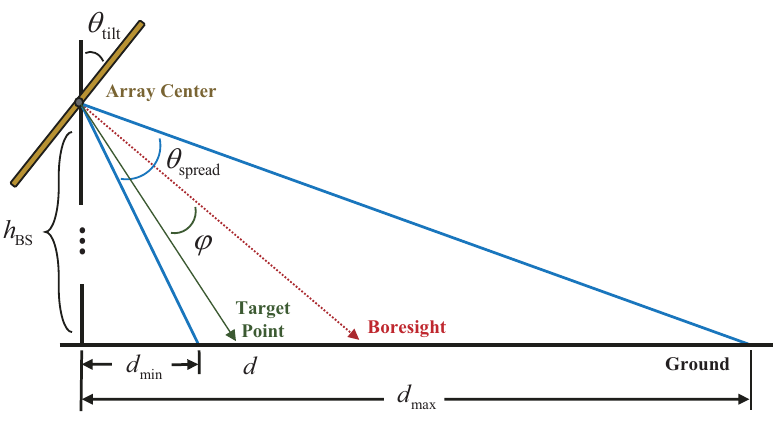}
\caption{Illustration of the BS broadcasting geometry. The array is tilted by $\theta_{\rm{tilt}}$, with height $h_{\rm{BS}}$ above ground. The vertical beam spread is $\theta_{\rm{spread}}$, assumed to be centered at the array boresight.}\label{model}
\vspace{-8pt}
\end{figure}

As shown in Fig. \ref{model}, we assume that the ground has a constant level, and the absolute height from the ground to the center of the BS's array is denoted as $h_{\rm{BS}}$. The array is tilted by an angle $\theta _{\rm{tilt}}$, known as the downtilt angle, to transmit signals to the ground. We use the term ``service region'' to denote the specific ground area that an array is designed to cover, as shown in Fig. \ref{region}. To broadcast public messages to all users in the service region, with a maximum distance of $d_{\rm{max}}$, the BS needs to generate a beam with a vertical beam spread of $\theta_{\rm{spread}}$, where we assume that the center of the beam aligns with the boresight of the array. Given the height of the base station and  the maximum distance of the service region, the downtilt angle and the beam spread have a mutually constraining relationship, as stated in Proposition 1.
\begin{proposition}\label{pp1}
Given the height of the base station $h_{\rm{BS}}$, the beam spread $\theta_{\rm{spread}}$, and the maximum distance of the service region ${d_{{\rm{max}}}}$, the downtilt angle $\theta _{\rm{tilt}}$ is given by
\begin{equation}\label{p9}
{\theta _{{\rm{tilt}}}} = \frac{{{\theta _{{\rm{spread}}}} + \pi }}{2} - \arctan \left( {\frac{{{d_{{\rm{max}}}}}}{{{h_{{\rm{BS}}}}}}} \right).    
\end{equation}
Given the height of the base station $h_{\rm{BS}}$, the downtilt angle $\theta _{\rm{tilt}}$, and the maximum distance of the service region ${d_{{\rm{max}}}}$, the beam spread $\theta_{\rm{spread}}$ is given by
\begin{equation}\label{p10}
{\theta _{{\rm{spread}}}} = 2\left[ {\arctan \left( {\frac{{{d_{{\rm{max}}}}}}{{{h_{{\rm{BS}}}}}}} \right) + {\theta _{{\rm{tilt}}}}} \right] - \pi,
\end{equation}
provided that \(\sin\theta_{\rm{spread}}+2/N\le 1\).
\end{proposition}
\begin{IEEEproof}
See Appendix A.
\end{IEEEproof}
\begin{remark}
Proposition 1 offers useful insights for cellular network planning and deployment. If the network coverage falls short of expectations, engineers can leverage (\ref{p9}) to tweak the downtilt angle, thereby achieving precise broadcast coverage. In cases where adjusting the downtilt angle is not an option, (\ref{p10}) comes into play, enabling the adjustment of the beam spread to fulfill the same purpose.\footnote{These methods operate under the assumption that the BS's transmission power is ample. Any deficiencies in coverage are attributed to the orientation of the array and the design of the beam.}    
\end{remark}

 \begin{figure}[t]
\centering
\includegraphics[width=2.7in]{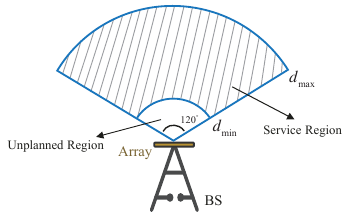}
\caption{Illustration of the service region and unplanned region. The BS is equipped with an antenna array covering a $120^\circ$ horizontal sector. The service region (gray) denotes the intended main‑lobe coverage area, bounded by the minimum distance $d_{\min}$ and the maximum distance $d_{\max}$ from the BS. The unplanned region in close proximity to the BS is served by the sidelobes.}\label{region}
\vspace{-12pt}
\end{figure}

In this work, we consider a cellular scenario where the geographical area is divided into multiple regions referred to as \emph{cells}, each served by a BS. To enhance spectrum utilization and reduce interference, each cell is usually divided into three sectors, each served by an antenna array of the base station. Hence, the horizontal coverage of the array is $120^{\circ}$, as shown in Fig. \ref{region}. The vertical coverage angle (known as the beam spread $\theta_{\rm{spread}}$ in this paper) is typically smaller, as the signal needs to be concentrated within the height range of ground users. If characterized by the azimuth angle and serving distance, the service region can be expressed by
\begin{equation}
  {{\cal C}_s} = \left\{ {\left( {\psi ,{d}} \right)\left| \;{\psi  \in \left[ { - \frac{\pi }{3},\frac{\pi }{3}} \right],{d} \in [{d_{{\rm{min}}}},{d_{{\rm{max}}}}]} \right.} \right\},
\end{equation}
while the unplanned region in the cell can be expressed by
\begin{equation}
      {{\cal C}_u} = \left\{ {\left( {\psi ,{d}} \right)\left| \;{\psi  \in \left[ { - \frac{\pi }{3},\frac{\pi }{3}} \right],{d} \in [0,{d_{{\rm{min}}}}]} \right.} \right\}.
\end{equation}

\begin{remark}
In practical cellular network deployments, signal coverage in unplanned regions is generally guaranteed. These regions are in close proximity to the BS, resulting in much lower propagation loss. Thus, by utilizing the sidelobe radiation of the beam or the electromagnetic waves reflected in the environment, users in the unplanned regions can be served.

We also note that near-field effects are negligible in our configuration. Specifically, for most UPA configurations at centimeter-wave bands, the near-field boundary is only a few meters. Over such short distances, path loss is minimal, and beam misalignment due to near-field propagation does not lead to significant energy loss. Consequently, our beamforming approach, which targets far-field cell-edge users, remains effective without requiring near-field modifications.
\end{remark}

\subsection{Propagation Model}
\begin{figure}[t]
\centering
\includegraphics[width=3.3in]{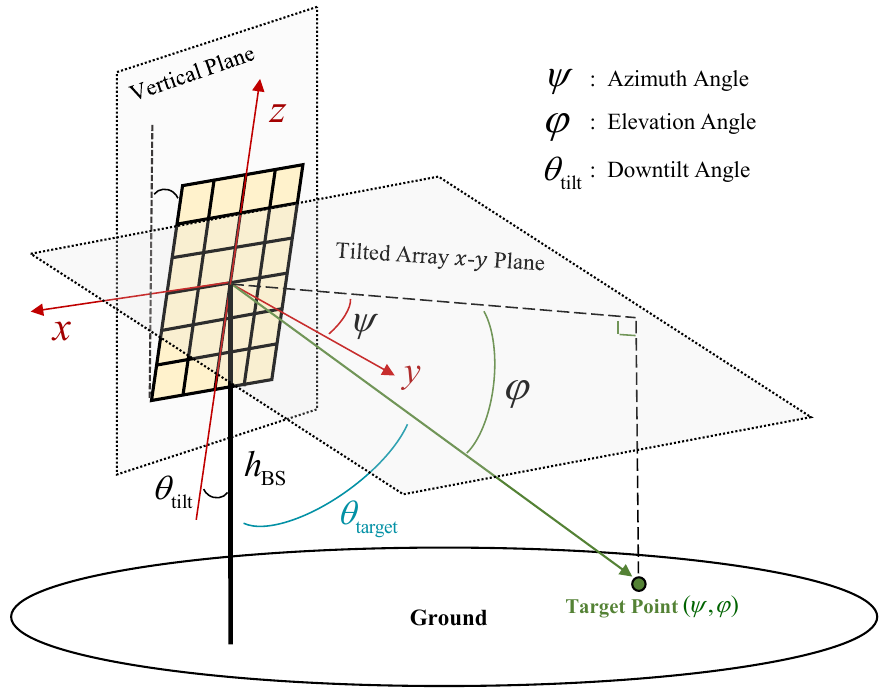}
\caption{Geometry of a target point as seen from the tilted array. The array is tilted by $\theta_{\rm{tilt}}$ with respect to the vertical plane. The BS height is $h_{\rm{BS}}$, and the target point on the ground is defined by azimuth $\psi$ and elevation $\varphi$ according to the tilted array $x$-$y$ plane. The angle $\theta_{\rm{target}}$ is the deviation of the target direction from the vertical axis.}\label{distance}
\vspace{-10pt}
\end{figure}
In conventional beam design, the beam pattern is specified in the array's local polar coordinates $(\psi, \varphi)$ without explicit consideration of the propagation environment. However, our goal is to design an \emph{environment-aware} broadcast beam that compensates for non-uniform path loss across the coverage area. To do so, we first need a model that maps each direction $(\psi, \varphi)$ from the tilted array to the actual propagation loss experienced by a ground user.

Consider a user located at direction $(\psi, \varphi)$ within the service region under LoS conditions. The actual received signal power differs from the radiated power in (\ref{dpower}) due to propagation loss. Following 3GPP TR 38.901, the path loss for an urban macro (UMa) scenario is given by
\begin{equation}
\begin{aligned}
 P_L^{\rm{dB}} &= 28.0 + 22{\log _{10}}({d_{{\rm{3D}}}}) + 20{\log _{10}}({f_c}),
\end{aligned}
\end{equation}
where $f_c$ denotes the center frequency normalized by $1$ GHz, and ${d_{{\rm{3D}}}}$ is the distance between the array center and the target point, whose value is normalized by $1$ m. The absolute magnitude loss due to the propagation can be written as 
\begin{equation}\label{ml0}
{M_L} = {10^{P_L^{{\rm{dB}}}/20}}.   
\end{equation}
As shown in Fig. \ref{distance}, the distance between the center of the array and the target point, ${d_{{\rm{3D}}}}$, is related to the height of the base station and the angle between the BS pole and the target direction, i.e., $h_{\rm{BS}}$ and $\theta _{\rm{target}}$, as stated in Proposition 2.
\begin{proposition}
When an array with a downtilt angle of $\theta _{\rm{tilt}}$ transmits a signal in the direction of azimuth angle $\psi$ and elevation angle $\varphi$ within the array's coordinate system, the propagation distance from the array to the point where the signal reaches the ground is
\begin{equation}
{d_{3D}} = \frac{{{h_{{\rm{BS}}}}}}{{\cos \psi \cos \varphi \sin {\theta _{{\rm{tilt}}}} + \sin \varphi \cos {\theta _{{\rm{tilt}}}}}}.
\end{equation}
\end{proposition}
\begin{IEEEproof}
See Appendix B.
\end{IEEEproof}
Proposition 2 implies that the propagation distance of the signal is related to both the azimuth angle and the elevation angle. The propagation distance, in turn, is associated with the signal propagation loss. Thus, one can express the absolute magnitude loss on propagation as a function of $\psi$ and $\varphi$, i.e.,
\begin{align}\label{ml}
&{M_L}(\psi ,\varphi ) = {10^{1.4}} \times {10^{1.1{{\log }_{10}}\left( {{h_{{\rm{BS}}}}} \right)}} \notag\\
 &\qquad \div {10^{1.1{{\log }_{10}}\left( {\cos \psi \cos \varphi \sin {\theta _{{\rm{tilt}}}} + \sin \varphi \cos {\theta _{{\rm{tilt}}}}} \right)}}\times {f_c}\\
 &\quad= l_{\rm{f}}{f_c} \times {{l_{\rm{d}}}^{{{\log }_{10}}\left[ \frac{{{h_{{\rm{BS}}}}}}{{\cos \psi \cos \varphi \sin {\theta _{{\rm{tilt}}}} + \sin \varphi \cos {\theta _{{\rm{tilt}}}}}} \right]}},\notag
\end{align}      
where $l_{\rm{f}}\approx 25.1$ and $l_{\rm{d}}\approx 12.6$ in the model of 3GPP UMa. The values of $l_{\rm{f}}$ and $l_{\rm{d}}$ should be properly adjusted for other wireless scenarios based on the statistical evaluation.
\begin{figure*}[t]
\centering
\includegraphics[width=7in]{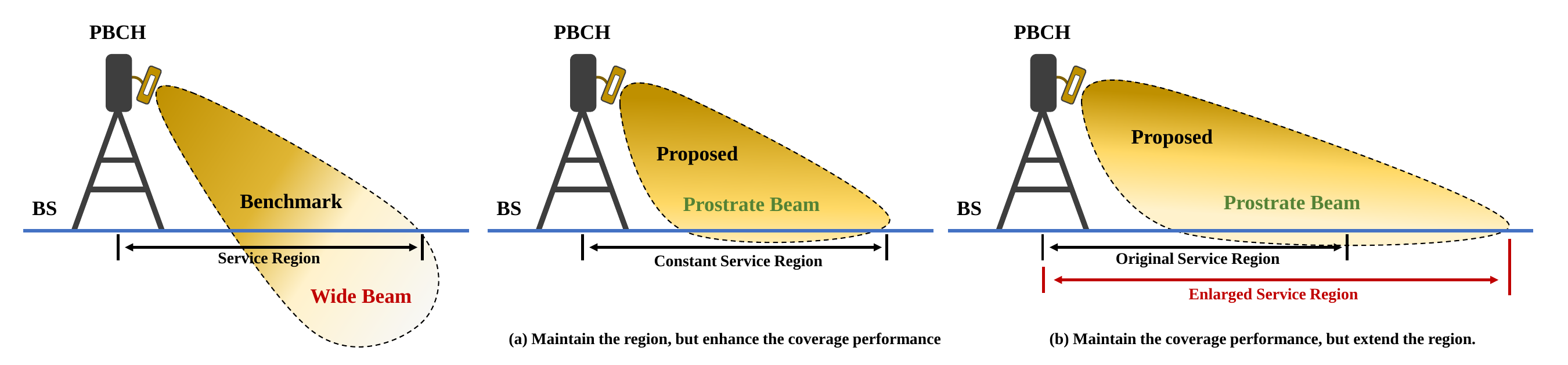}
\caption{Illustration of the expected beam patterns of the wide-beam benchmark and our proposed prostrate-beam designs.}\label{propose}
\vspace{-6pt}
\end{figure*}
\section{Motivation and Beam-Shaping Problem}
In practical systems, to ensure broadcast signal coverage across the service region, it is common to design two uniform wide beams, one in the horizontal direction and the other in the vertical direction, to achieve a 3D coverage\cite{3dcov}. This method is considered the \emph{benchmark}.

\subsection{Motivation for Optimizing the Benchmark}

\begin{figure}[t]
\centering
\includegraphics[width=3.4in]{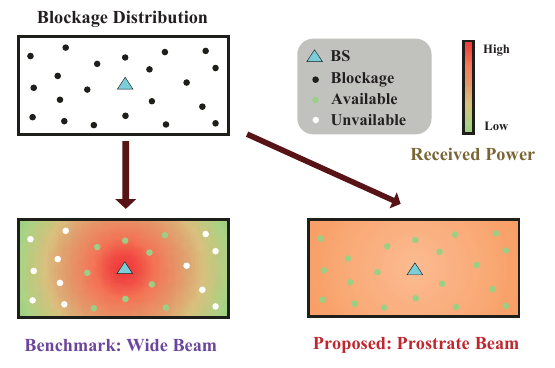}
\caption{Benchmark and proposed schemes in blockage scenarios.}\label{blockage}
\vspace{-6pt}
\end{figure}
For any beamforming vector, the beamforming weights in antenna domain and the radiation response in beam domain constitute a pair of discrete-time Fourier transform (DTFT) transformation. Let us define a pair of virtual angles to characterize the vertical direction and azimuth direction by
\begin{equation}
 {\varphi _v} = \sin \varphi,\;\;{\psi _h} = \cos \varphi \sin \psi   
\end{equation}
with ${\varphi _v},{\psi _h} \in [ - 1,1]$. Then, the beam radiation in the direction of $({\varphi _v},{\psi _h})$ can be rewritten as
\begin{equation}
\hat r({\psi _h},{\varphi _v}) = \sum\nolimits_{n,m} {{{\left[ {{{{\bf{\hat A}}}^*}({\psi _h},{\varphi _v}) \circ {\bf{W}}} \right]}_{n,m}}} 
\end{equation}
where
\begin{equation}
\begin{aligned}
&{\left[ {{\bf{\hat A}}(\psi _h,\varphi _v)} \right]_{n,m}} = \frac{1}{{\sqrt {NM} }}{e^{j\pi (m{\psi _h} + n{\varphi _v})}},\\
&\; n=0,1,...,N-1,\;\;m=0,1,...,M-1.
\end{aligned}
\end{equation}
Then, the \emph{Spatial-Angular Parseval relation} is given by\footnote{A proof is provided in Appendix C for interested readers.}
\begin{equation}\label{SAP}
{\int_{ - 1}^1 {\int_{ - 1}^1 {\left| {\hat r({\psi _h},{\varphi _v})} \right|}^2 } }d{\psi _h}d{\varphi _v} = \frac{4}{MN}{\sum\limits_{n = 0}^{N - 1} {\sum\limits_{m = 0}^{M - 1} {\left| {{{\bf{W}}_{n,m}}} \right|}^2 }} .  
\end{equation}
where ${\varphi _v} = \sin \varphi$ and ${\psi _h} = \cos \varphi \sin \psi$. This indicates that when the total transmission power is fixed, the integral of the beam radiation power remains constant. For the traditional wide beam, the radiation power is uniform within the main lobe, as shown in Fig. \ref{propose}. However, the propagation loss at close range is lower than that at far range, resulting in the received signal power not being evenly distributed across the service region. In particular, the received signal power close to the transmitter may far exceed the 
broadcasting demodulation threshold of the users, e.g., the minimum required reference signal received power (RSRP). Thus, the conventional wide beam may result in imbalanced received power within the service region. By measuring the synchronization signal RSRP (SS-RSRP) over the service region using a 5G UE in a real-world base station scenario, it can be confirmed that this issue does indeed exist under conventional broadcast beams, as presented in Figure 6 of \cite{UErsrp}. 

Motivated by this, one can lower the beam gain at large elevation angles to increase the beam gain at small elevation angles, thereby raising the average received power in the service region (see Fig. \ref{propose}(a)), or make the public message effectively propagate further to extend the service region (see Fig. \ref{propose}(b)). In scenarios where blockages exist or in O2I scenarios, blockages are assumed to be \emph{uniformly distributed} in space following a Poisson point process\cite{Ruiz2021}. To ensure that all blocking points can access the broadcast signal while maximizing the size of the service region, the most effective method is to evenly distribute the received power within the region, as depicted in Fig. \ref{blockage}. This method involves reshaping the conventional wide beam into a prostrate beam, which effectively follows a \emph{Robin Hood} principle, taking excess signal power from the vicinity of the base station and allocating it to distant areas, thus optimizing the coverage profile. It can be demonstrated that, within the same service region, if the proposed approach fails to serve/cover all blocked users, then the benchmark approach will likewise fail to do so.

Following the aforementioned ideas, we aim to achieve a prostrate beam, whose radiation power reaching the ground appears to crawl along the ground. In particular, the prostrate beam proposed mainly requires changes to the benchmark in the vertical direction, i.e., the beam pattern in the vertical direction is no longer a uniform wide beam.\footnote{For ease of exposition, in the following content, we merely discuss the vertical beam design. In Section \ref{secnu}, we will demonstrate that our proposed method can also be applied to 3D beams. }

\subsection{Problem Formulation}
Beam shaping technology first requires a predefined target beam pattern $\mathcal{T}(\psi,\varphi)$, which specifies the magnitude of desired beam radiation in different directions. Then, the general beam shaping problem can be formulated as
\[
({\rm{P}}0): \mathop {\min }\limits_{{{\bf{w}}_v},{{\bf{w}}_h}} \int\limits_\psi  {\int\limits_\varphi  {{{\Big| \; |r(\psi ,\varphi )|-{{\cal T}(\psi ,\varphi ) }\Big|}^2}d\varphi d\psi } },\]
where $|r(\psi ,\varphi )|$ is the achievable magnitude of beam radiation given in (\ref{radr}). The beam shaping technology itself is independent of any specific target beam pattern. Additionally, the 3D beam shaping problem can be reduced to a 2D one if we separately design ${{\bf{w}}_v}$ or ${{\bf{w}}_h}$.  For the case study in this work, we consider shaping a 2D prostrate beam, where the beam gain exactly compensates for the propagation loss, thereby ensuring the received signal power is uniform across the entire service region, i.e.,
\begin{equation}
\mathcal{T}(\varphi ) = \left\{ {\begin{array}{*{20}{c}}
{{M_L}(\varphi ),\;\;\;{\rm{if}}\;\varphi \in {\Pi _c}},\\
{\mu (\varphi ),\;\;\;\;\;{\rm{if}}\;\varphi  \in {\Pi _t},}\\
{\;\;0,\;\;\qquad  {\rm{if}}\;\varphi  \in {\Pi _s},}
\end{array}} \right.    
\end{equation}
where ${\Pi _c}$, ${\Pi _t} = {\Pi _t^+} \cup {\Pi _t^-}$, and ${\Pi _s}$  denote the coverage range, transition range, and stop range,{\footnote{The terms \emph{pass range}, \emph{transition range}, and \emph{stop range} correspond to the \emph{pass band}, \emph{transition band}, and \emph{stop band} for filter design, respectively.}} respectively. ${M_L}(\varphi )$, $\mu (\varphi )$, and $0$ are the absolute value of target beam pattern in the coverage range, transition range, and stop range, respectively. The ranges mentioned are defined as
\begin{align}\label{range}
{\Pi _c} &= \left[ { - \frac{1}{2}{\theta _{{\rm{spread}}}},\frac{1}{2}{\theta _{{\rm{spread}}}}} \right], \notag\\
\Pi _t^ +  &= \left[ {\frac{1}{2}{\theta _{{\rm{spread}}}},\frac{1}{2}\theta _{{\rm{spread}}}^*} \right], \notag \\
\Pi _t^ -  &= \left[ { - \frac{1}{2}\theta _{{\rm{spread}}}^*, - \frac{1}{2}{\theta _{{\rm{spread}}}}} \right], \notag\\
{\Pi _s} &= \left[ { - \frac{\pi }{2}, - \frac{1}{2}\theta _{{\rm{spread}}}^ * } \right] \cup \left[ {\frac{1}{2}\theta _{{\rm{spread}}}^ * ,\frac{\pi }{2}} \right],
\end{align}
where ${\theta _{{\rm{spread}}}^ * }$ has an additional spatial resolution compared to ${\theta _{{\rm{spread}}} }$, which can be obtained by
\begin{equation}
\theta _{{\rm{spread}}}^ *  = \arcsin \left( {\sin {\theta _{{\rm{spread}}}} + \frac{2}{N}} \right).    
\end{equation}
In particular, $\mu (\varphi ) \in \mathbb{R}$ is an auxiliary function with $0 \leq \mu (\varphi ) \leq{M_L}(\varphi ^+)$ in the transition range $\Pi _t^+$, and $0 \leq \mu (\varphi ) \leq{M_L}(\varphi ^-)$ in the transition range $\Pi _t^-$, where $\varphi ^+ = \frac{1}{2}\theta _{{\rm{spread}}}$, $\varphi ^- = -\frac{1}{2}\theta _{{\rm{spread}}}$. Define a function to characterize the loss of beam radiation compared to the target beam pattern as 
\allowdisplaybreaks
\begin{align}
\varepsilon \left( {\bf{w}}_v \right) &\buildrel \Delta \over = \int\limits_{\varphi  \in {\Pi _c}} {\left| {\left| {{{\bf{a}}_v}{{(\varphi )}^H}{\bf{w}}_v} \right| - {M_L}(\varphi )} \right|^2} d\varphi \notag\\
& \quad  + \mathop {\min }\limits_{\mu}\int\limits_{\varphi  \in {\Pi _t} } {\left| {\left| {{{\bf{a}}_v}{{(\varphi )}^H}{\bf{w}}_v} \right| - \mu (\varphi )} \right|^2} d\varphi \\
& \qquad \qquad \qquad + \int\limits_{\varphi  \in {\Pi _s}} {\left| {\left| {{{\bf{a}}_v}{{(\varphi )}^H}{\bf{w}}_v} \right| - 0} \right|^2} d\varphi,\notag
\end{align}
Then, the beam shaping problem for creating a prostrate beam can be formulated as 
\[
\begin{aligned}
({\rm{P}}1): \qquad &\mathop {\min }\limits_{{{\bf{w}}_v}, \mu (\cdot ) } \varepsilon \left( {{{\bf{w}}_v}} \right)\\
&{\rm{s}}.{\rm{t}}.\;{\rm{0}} \le \mu (\varphi ) \le {M_L}(\varphi^+ ),\;\varphi  \in {\Pi _t^+},\\
&\qquad{\rm{0}} \le \mu (\varphi ) \le {M_L}(\varphi^- ),\;\varphi  \in {\Pi _t^-},
\end{aligned}    
\]
in which 
\begin{equation}
{M_L}(\varphi ) = l_f{f_c} \times {{l_d}^{{{\log }_{10}}\left[ \frac{{h_{{\rm{BS}}}}}{\sin (\varphi  + {\theta _{{\rm{tilt}}}})} \right]}}    
\end{equation}
is the magnitude loss without considering $\psi$. In the problem for 3D prostrate-beam design, the value of ${M_L}$ is also influenced by azimuth angle $\psi$, as illustrated in (\ref{ml}). It is observed that the problem formulated does not take the transmit power of the beamforming vector into account. This is because the optimized solution can meet power constraints by scaling up. Furthermore, this scaled-up solution also ensures uniform received power at any location within the service region.

\section{Beam Shaping Optimization}
We propose a new prostrate beam design in this section by solving problem $(\rm{P}1)$. Specifically, we first discretize the objective of the original problem $(\rm{P}1)$ to generate a new problem $(\rm{P}2)$ in terms of matrix representation. Then, we demonstrate through derivation that the problem $(\rm{P}2)$ is equivalent to a TPO problem $(\rm{P}3)$ that optimizes the phases of the target beam pattern.  To formulate a convex optimization, we transform the original phase variables into a vector variable and incorporate their elements into quadratic constraints. By lifting the vector variable to a positive semidefinite matrix and introducing rank-one equivalent constraints, we can iteratively solve the problem by decomposing it into several convex sub-problems using successive convex approximation (SCA) technology.
\subsection{Matrix Representation}
To express the original objective function in a compact matrix form, we define $O$ angle sampling points within $[-\pi/2,\pi/2]$ to characterize the beam pattern in elevation as
\begin{equation} \label{v2}
 {\varphi _o} =  - \frac{\pi }{2} + \frac{{\left( {o - 1} \right)\pi }}{{O - 1}},\;\;o = 1,2,...,O.   
\end{equation}
Then, the objective function of $(\rm{P}1)$ can be approximated as
\allowdisplaybreaks
\begin{align}
\varepsilon_D(\mathbf{w}_v) = & \sum_{\varphi_o \in \Pi_c} \left| \left| \mathbf{a}_v(\varphi_o)^H \mathbf{w}_v \right| - M_L(\varphi_o) \right|^2 \Delta \varphi \nonumber \\
& + \min_{\{\mu_o\}_{o=1}^O} \; \sum_{\varphi_o \in \Pi_t} \left| \left| \mathbf{a}_v(\varphi_o)^H \mathbf{w}_v \right| - \mu_o \right|^2 \Delta \varphi \nonumber \\
& + \sum_{\varphi_o \in \Pi_s} \left| \left| \mathbf{a}_v(\varphi_o)^H \mathbf{w}_v \right| - 0 \right|^2 \Delta \varphi ,
\end{align}
where $\Delta \varphi  = \pi /\left( {O - 1} \right)$ and $\mu _o =\mu (\varphi _o)$. This objective contains two modulus operations, which is difficult to handle directly with existing optimization methods. Therefore, we propose introducing $O$ auxiliary variables, i.e., $\{\kappa_o\}_{o=1}^O$ and constructing a new objective as
\allowdisplaybreaks
\begin{align}
&{{\tilde \varepsilon }_D}\left( {{{\bf{w}}_v}},\{\kappa_o\}_{o=1}^O \right) \notag\\
&= \sum\limits_{{\varphi _o} \in {\Pi _c}} {{{\left| {{{\bf{a}}_v}{{({\varphi _o})}^H}{{\bf{w}}_v} - {M_L}({\varphi _o}){e^{j{\kappa _o}}}} \right|}^2}} \Delta \varphi \notag\\
& \qquad \quad  + \mathop {\min }\limits_{\{\mu _o\}_{o=1}^O}\sum\limits_{{\varphi _o} \in {\Pi _t}} {{{\left| {{{\bf{a}}_v}{{({\varphi _o})}^H}{{\bf{w}}_v} - \mu _o{e^{j{\kappa _o}}}} \right|}^2}} \Delta \varphi \\
& \qquad \qquad \qquad + \sum\limits_{{\varphi _o} \in {\Pi _s}} {{{\left| {{{\bf{a}}_v}{{({\varphi _o})}^H}{{\bf{w}}_v} - 0 \cdot {e^{j{\kappa _o}}}} \right|}^2}} \Delta \varphi .\notag
\end{align}  
It is easy to verify that the following minimization problems are equivalent:
\[\mathop {\min }\limits_{{{\bf{w}}_v}} {\varepsilon _D}\left( {{{\bf{w}}_v}} \right) \quad \Leftrightarrow \mathop {\min }\limits_{{{\bf{w}}_v},\{ {\kappa _o}\} _{o = 1}^O } {{\tilde \varepsilon }_D}\left( {{{\bf{w}}_v},\{ {\kappa _o}\} _{o = 1}^O } \right),\]
because the optimal $\kappa _o$ is the phase of ${{\bf{a}}_v}{{({\varphi _o})}^H}{{\bf{w}}_v}$. To recast the original problem in terms of matrix representation, we need further define an $N \times O$ matrix whose column consists of ${{\bf{a}}_v}{({\varphi _o})}$ for all $o = 1, 2,...,O$, i.e.,
\begin{equation}
{\bf{V}} \buildrel \Delta \over = \left[ {{{\bf{a}}_v}({\varphi _1}),{{\bf{a}}_v}({\varphi _2}),...,{{\bf{a}}_v}({\varphi _O})} \right].
\end{equation}
Next, we define a target-pattern function vector ${\bf{t}}( \{ {\kappa _o, {\mu} _o}\} _{o = 1}^O )$ that consists of the target complex values compared to ${{\bf{a}}_v}{{({\varphi _o})}^H}{{\bf{w}}_v}$ in ${{\tilde \varepsilon }_D}$, in which the $o$-th element
\begin{equation}\label{to3}
{{\bf{t}}}[o] = \left\{ {\begin{aligned}
&{{M_L}({\varphi _o}){e^{j{\kappa _o}}},\;{\rm{if}}\;{\varphi _o} \in {\Pi _c}},\\
&{\mu _o{e^{j{\kappa _o}}},\qquad \;\;{\rm{if}}\;{\varphi _o} \in {\Pi _t}},\\
&{0 \cdot {e^{j{\kappa _o}}},\qquad \; {\rm{if}}\;{\varphi _o} \in {\Pi _s}}.
\end{aligned}} \right.
\end{equation}
By discarding the constant $\Delta \varphi$ and merging the variable function $\mu$, the original optimization problem $(\rm{P}1)$ can be alternatively expressed as 
\[\begin{aligned}
({\rm{P2}}):\quad & \mathop {\min }\limits_{{{\bf{w}}_v},\;\{ {\kappa _o, {\mu} _o}\} _{o = 1}^O  } \left\| {{{\bf{V}}^H}{{\bf{w}}_v} - {\bf{t}}\left( {\{ {\kappa _o, {\mu} _o}\} _{o = 1}^O } \right)} \right\|_F^2\\
&{\rm{s}}.{\rm{t}}.\;0 \le \mu _o \le {M_L}(\varphi^+ ),\;\varphi _o \in {\Pi _t^+},\\
&\qquad 0 \le \mu _o \le {M_L}(\varphi^- ),\;\varphi _o \in {\Pi _t^-}.
\end{aligned}\]
\subsection{Equivalent TPO Problem}
One can notice that given $\{ {\kappa _o, {\mu} _o}\} _{o = 1}^O $, the optimal solution of ${{\bf{w}}_v}$ can be obtained in closed form as 
\begin{equation} 
{\bf{w}}_v^{{\rm{opt}}} = {\left( {{\bf{V}}{{\bf{V}}^H}} \right)^{ - 1}}{\bf{Vt}}\left( {\{ {\kappa _o, {\mu} _o}\} _{o = 1}^O } \right), 
\end{equation}
where we assume that $O > N$ and thus ${{\bf{V}}{{\bf{V}}^H}}$ is invertible. By substituting this solution into the problem $(\rm{P}2)$, the remaining task is to optimize the phases of the target pattern in all ranges, i.e., $\{ {\kappa _o}\} _{o = 1}^O $, as well as the magnitudes of that in transition range, i.e., $\{ \mu _o\;|\;{\varphi _o} \in {\Pi _t}\} $. We refer to this kind of issue as a TPO problem, i.e.,
\allowdisplaybreaks
\begin{align*}
({\rm{P3}}):& \mathop {\min }\limits_{\{ {\kappa _o, {\mu} _o}\} _{o = 1}^O  } \left\| {\left[ {{{\bf{V}}^H}{{\left( {{\bf{V}}{{\bf{V}}^H}} \right)}^{ - 1}}{\bf{V}} - {\bf{I}}} \right]{\bf{t}}\left( {\{ {\kappa _o, {\mu} _o}\} _{o = 1}^O } \right)} \right\|_F^2\\
&\quad {\rm{s}}.{\rm{t}}.\;{\rm{0}} \le \mu _o \le {M_L}(\varphi^+ ),\;\varphi _o  \in {\Pi _t^+},\\
&\quad \qquad{\rm{0}} \le \mu _o \le {M_L}(\varphi^- ),\;\varphi _o  \in {\Pi _t^-}.
\end{align*}    
Since ${{{\bf{V}}^H}{{\left( {{\bf{V}}{{\bf{V}}^H}} \right)}^{ - 1}}{\bf{V}}}$ is a Hermitian matrix, i.e., 
\begin{equation}
{{\bf{V}}^H}{\left( {{\bf{V}}{{\bf{V}}^H}} \right)^{ - 1}}{\bf{V}} = {\left[ {{{\bf{V}}^H}{{\left( {{\bf{V}}{{\bf{V}}^H}} \right)}^{ - 1}}{\bf{V}}} \right]^H},    
\end{equation}
the objective function of $(\rm{P}3)$ can be rewritten as
\begin{align}
&\left\| {\left[ {{{\bf{V}}^H}{{\left( {{\bf{V}}{{\bf{V}}^H}} \right)}^{ - 1}}{\bf{V}} - {\bf{I}}} \right]{\bf{t}}\left( {\{ {\kappa _o, {\mu} _o}\} _{o = 1}^O } \right)} \right\|_F^2\\
 &= {\bf{t}}{\left( {\{ {\kappa _o, {\mu} _o}\} _{o = 1}^O } \right)^H}{{\bf{V}}^H}{\left( {{\bf{V}}{{\bf{V}}^H}} \right)^{ - 1}} \times {\bf{V}}{{\bf{V}}^H}{\left( {{\bf{V}}{{\bf{V}}^H}} \right)^{ - 1}}  \notag\\
& \;\; \times {\bf{Vt}}{\left( {\{ {\kappa _o, {\mu} _o}\} _{o = 1}^O } \right)^H} - 2{\bf{t}}{\left( {\{ {\kappa _o, {\mu} _o}\} _{o = 1}^O } \right)^H}{{\bf{V}}^H}{\left( {{\bf{V}}{{\bf{V}}^H}} \right)^{ - 1}} \notag\\
& \quad \times {\bf{Vt}}{\left( {\{ {\kappa _o, {\mu} _o}\} _{o = 1}^O } \right)} + {\bf{t}}{\left( {\{ {\kappa _o, {\mu} _o}\} _{o = 1}^O } \right)^H}{\bf{t}}\left( {\{ {\kappa _o, {\mu} _o}\} _{o = 1}^O } \right) \notag\\
& = {\bf{t}}{\left( {\{ {\kappa _o, {\mu} _o}\} _{o = 1}^O } \right)^H}{\bf{t}}\left( {\{ {\kappa _o, {\mu} _o}\} _{o = 1}^O } \right) \notag\\
& \qquad- {\bf{t}}{\left( {\{ {\kappa _o, {\mu} _o}\} _{o = 1}^O } \right)^H}{{\bf{V}}^H}{\left( {{\bf{V}}{{\bf{V}}^H}} \right)^{ - 1}}{\bf{Vt}}{\left( {\{ {\kappa _o, {\mu} _o}\} _{o = 1}^O } \right)^H}. \notag
\end{align}
By defining the matrix 
\begin{equation}
 {\bf{B}} = {\bf{I}} - {{\bf{V}}^H}{\left( {{\bf{V}}{{\bf{V}}^H}} \right)^{ - 1}}{\bf{V}},   
\end{equation}
the problem $(\rm{P}3)$ can be simplified as 
\[
\begin{aligned}
({\rm{P4}}):\;\;&\mathop {\min }\limits_{\{{ {\kappa _o, \mu _o}\} _{o = 1}^O}  } {\bf{t}}{\left( {\{ {\kappa _o, {\mu} _o}\} _{o = 1}^O } \right)^H}{\bf{Bt}}\left( {\{ {\kappa _o, {\mu} _o}\} _{o = 1}^O } \right)\\
&\quad {\rm{s}}.{\rm{t}}.\;{\rm{0}} \le \mu _o \le {M_L}(\varphi^+ ),\;\varphi _o \in {\Pi _t^+},\\
&\quad \qquad{\rm{0}} \le \mu _o \le {M_L}(\varphi^- ),\;\varphi _o \in {\Pi _t^-}.
\end{aligned}
\]
It can be observed that problem (P4) is a simplified TPO problem with a quadratic objective, and all its constraints are convex functions. In the TPO problem, we only need to consider how to optimize the parameters of the target-pattern function vector $\{{\kappa _o, {\mu} _o}\} _{o = 1}^O $. Although the objective function is convex with respect to the target-pattern function vector, it is non-convex with respect to the parameters to be optimized. To solve for these parameters, we try to construct an optimization problem for the target-pattern function vector and incorporate the constant modulus property of its elements as constraints.

\subsection{Convexification Attempts}
Replacing the function vector ${\bf{t}}\left( {\{ {\kappa _o, {\mu} _o}\} _{o = 1}^O } \right)$  with a variable vector  $\bf{t}$, and treating the function structure in (\ref{to3}) as quadratic constraints on the elements of the variable vector $\bf{t}$, we can reformulate $({\rm{P4}})$ as a quadratically
constrained quadratic program (QCQP):
\[
\begin{aligned}
({\rm{P5}}):\quad & \mathop {\min }\limits_{\bf{t}} {{\bf{t}}^H}{\bf{Bt}}\\
&{\rm{s}}.{\rm{t}}.\;{\left| {{{\bf{t}}_i}} \right|^2} = {M_L}{({\varphi _i})^2}\;{\rm{with}}\;{\varphi _i} \in {\Pi _c},\\
&\quad \;\;\; {\left| {{{\bf{t}}_j}} \right|^2} \le {M_L}{({\varphi ^+})^2}\;{\rm{with}}\;{\varphi _j} \in {\Pi _t^+},\\
&\quad \;\;\; {\left| {{{\bf{t}}_j}} \right|^2} \le {M_L}{({\varphi ^-})^2}\;{\rm{with}}\;{\varphi _j} \in {\Pi _t^-},\\
&\quad \;\;\; {\left| {{{\bf{t}}_k}} \right|^2} = 0\;{\rm{with}}\;{\varphi _k} \in {\Pi _s},
\end{aligned}
\]
where the original variables in $(\rm{P}4)$ are implicit in the variable vector $\bf{t}$. To reformulate the constraints into convex forms, we lift the variable vector into a positive semidefinite matrix, i.e.,
\begin{equation}
{\bf{T}} = {\bf{t}}{{\bf{t}}^H}\;\;{\rm{and}}\;\; {\rm{rank}}({\bf{T}})=1.    
\end{equation}
In particular, the rank-one constraint is equivalent to
\begin{equation}
{\left\| {\bf{T}} \right\|_*} - {\left\| {\bf{T}} \right\|_2} = 0,    
\end{equation}
where ${\left\| \cdot \right\|_*}$ represents the nuclear norm which equals to the summation of all singular value and ${\left\| \cdot \right\|_2}$ represents the spectral norm which equals to the largest singular value. By imposing this constraint as a penalty with a weight of $\rho$, we can reformulate $(\rm{P}5)$ as
\begin{align*}
({\rm{P6}}):\;&\mathop {\min }\limits_{\bf{T}} {\rm{Tr}}\left( {{\bf{TB}}} \right) + \rho \left[ {{\rm{Tr}}\left( {\bf{T}} \right) - {{\left\| {\bf{T}} \right\|}_2}} \right]\\
&{\rm{s}}.{\rm{t}}.\;\;{{\bf{T}}_{i,i}} = {M_L}{({\varphi _i})^2}\;{\rm{with}}\;{\varphi _i} \in {\Pi _c},\\
&\qquad  {{\bf{T}}_{j,j}} \le {M_L}{({\varphi ^+})^2}\;{\rm{with}}\;{\varphi _j} \in {\Pi _t^+},\\
&\qquad  {{\bf{T}}_{j,j}} \le {M_L}{({\varphi ^-})^2}\;{\rm{with}}\;{\varphi _j} \in {\Pi _t^-},\\
&\qquad {{\bf{T}}_{k,k}} = 0\;{\rm{with}}\;{\varphi _k} \in {\Pi _s},
\end{align*}

where we replace ${\left\| {\bf{T}} \right\|_*}$ with ${\rm{Tr}}\left( {\bf{T}} \right)$ since ${\bf{T}}$ is a positive semidefinite matrix. Since $\text{Tr}(\mathbf{T}) - \|\mathbf{T}\|_2$ is always greater than or equal to zero, it will act as a penalty by increasing the objective function when the rank-1 constraint is not satisfied. It can be observed that in the current optimization problem, all components in the objective function and constraint function have become convex, except for $\mathop {\min }\nolimits_{\bf{T}}  - {{\left\| {\bf{T}} \right\|}_2}$.

\subsection{SCA Decomposition}
Next, we use the SCA technique to decompose the problem  $(\rm{P}6)$  into several sub-problems and solve them iteratively.  The basic idea of SCA is to approximate the non-convex objective function with a convex function at each optimization step. Let $f({\bf{T}}) = {\rm{Tr}}\left( {{\bf{TB}}} \right) + \rho \left[ {{\rm{Tr}}\left( {\bf{T}} \right) - {{\left\| {\bf{T}} \right\|}_2}} \right]$. The steps of the SCA algorithm can be described as follows:
\begin{enumerate}
    \item Choose an initial point ${\bf{T}}_0$ in step 0.
    \item At step $k$, construct a convex approximation ${f_k}({\bf{T}};{{\bf{T}}_k})$ that matches the function value at ${\bf{T}}_k$, i.e., 
    \[
     {f_k}({\bf{T}};{{\bf{T}}_k})= f({{\bf{T}}_k}),
    \]
    and is an upper bound on $f({\bf{T}})$ in the neighborhood of ${\bf{T}}_k$.
    \item Solve the following convex optimization problem at step $k$:
     \[{{\bf{T}}_{k + 1}} = \arg \mathop {\min }\limits_{\bf{T}} {f_k}({\bf{T}};{{\bf{T}}_k}).\]
    \item  If a certain convergence criterion is met, stop. Otherwise, return to step 2.
\end{enumerate}
Following above description, we construct a convex approximation ${f_k}({\bf{T}};{{\bf{T}}_k})$ at step $k$ as
\begin{equation}
f_k\left( {{\bf{T}};{{\bf{T}}_k}} \right)  =  {\rm{Tr}}\left( {{\bf{TB}}} \right) + \rho {\rm{Tr}}\left( {\bf{T}} \right) - \rho {h_k}\left( {{\bf{T}};{{\bf{T}}_k}} \right),
\end{equation}
where ${{{\bf{T}}_k}}$ is a constant matrix obtained in the $(k-1)$-th optimization step. ${h_k}\left( {{\bf{T}};{{\bf{T}}_k}} \right)$ is the first-order Taylor approximation of ${{{\left\| {\bf{T}} \right\|}_2}}$, which can be written as
\begin{equation}
\begin{aligned}
{h_k}\left( {{\bf{T}};{{\bf{T}}_k}} \right) &= {\left\| {{{\bf{T}}_k}} \right\|_2} + {\mathop{\rm Re}\nolimits} \left( {{\rm{Tr}}\left[ {{\partial _{{{\bf{T}}_k}}}{{\left\| {\bf{T}} \right\|}_2}\left( {{\bf{T}} - {{\bf{T}}_k}} \right)} \right]} \right)\\
 &= {\left\| {{{\bf{T}}_k}} \right\|_2} + {\mathop{\rm Re}\nolimits} \left( {{\rm{Tr}}\left[ {{{\bf{q}}_k}{\bf{q}}_k^H\left( {{\bf{T}} - {{\bf{T}}_k}} \right)} \right]} \right)\\
 &= {\left\| {{{\bf{T}}_k}} \right\|_2} +   {{\bf{q}}_k^H\left( {{\bf{T}} - {{\bf{T}}_k}} \right){{\bf{q}}_k}},
\end{aligned}
\end{equation}
 where ${{\bf{q}}_k}$ is the singular vector corresponding to the largest singular value of ${{{\bf{T}}_k}}$, and the second equation holds due to the fact of ${\partial _{{{\bf{T}}_k}}}{\left\| {\bf{T}} \right\|_2} = {{\bf{q}}_k}{\bf{q}}_k^H$. Then, the solution obtained in the $k$-th optimization step is used for objective parameter in the $(k+1)$-th step, i.e., 
\begin{equation}\label{tk1}
{{\bf{T}}_{k + 1}} = \arg \mathop {\min }\limits_{\bf{T}} {\rm{Tr}}\left( {{\bf{TB}}} \right) + \rho {\rm{Tr}}\left( {\bf{T}} \right) - \rho {h_k}\left( {{\bf{T}};{{\bf{T}}_k}} \right)    .
\end{equation}
By discarding the constant term in (\ref{tk1}) and some matrix manipulations, the optimization sub-problem in the $k$-th step can be written as 
\begin{align}
({\rm{P6}})_k:\;&{{\bf{T}}_{k + 1}} = \arg \mathop {\min }\limits_{\bf{T}} {\rm{Tr}}\left( {{\bf{TC}}_k} \right)\nonumber\\
&{\rm{s}}.{\rm{t}}.\;{{\bf{T}}_{i,i}} = {M_L}{({\varphi _i})^2}\;{\rm{with}}\;{\varphi _i} \in {\Pi _c}\nonumber\\
&\qquad  {{\bf{T}}_{j,j}} \le {M_L}{({\varphi ^+})^2}\;{\rm{with}}\;{\varphi _j} \in {\Pi _t^+},\nonumber\\
&\qquad  {{\bf{T}}_{j,j}} \le {M_L}{({\varphi ^-})^2}\;{\rm{with}}\;{\varphi _j} \in {\Pi _t^-},\nonumber\\
&\qquad {{\bf{T}}_{k,k}} = 0\;{\rm{with}}\;{\varphi _k} \in {\Pi _k},\nonumber
\end{align}
where ${\bf{C}}_k = {\bf{B}} + \rho \left( {{\bf{I}} - {{\bf{q}}_k}{\bf{q}}_k^H} \right)$. Supposing that $R$ is the total number of iterations and ${{\bf{T}}_R} = {\bf{U\Lambda }}{{\bf{U}}^H}$ is the eigen-decomposition
of the solution to $(\rm{P}6)$, the latent prostrate-beam solution in $(\rm{P}1)$ can be recovered as 
\begin{equation}
{{\bf{w}}_v} = \beta {\left( {{\bf{V}}{{\bf{V}}^H}} \right)^{ - 1}}{\bf{VU}}\left[ {:,1} \right]\sqrt {{{\bf{\Lambda }}_{1,1}}},   
\end{equation}
in which  ${\bf{U}}\left[ {:,1} \right]$ is the first column of ${\bf{U}}$, ${{{\bf{\Lambda }}_{1,1}}}$ is the element in the first column and first row of ${{{\bf{\Lambda }}}}$, and $\beta$ is a power scaling factor that should be $1$ in the problem $(\rm{P}1)$. For prostrate beam with transmit power $P$, the power scaling vector can be set to 
\begin{equation}
\beta  = \frac{{\sqrt P }}{{{{\left\| {{{\left( {{\bf{V}}{{\bf{V}}^H}} \right)}^{ - 1}}{\bf{VU}}\left[ {:,1} \right]\sqrt {{{\bf{\Lambda }}_{1,1}}} } \right\|}_F}}}.
\end{equation}

\subsection{Overall Algorithm and Convergence Analysis}

\renewcommand{\algorithmicrequire}{\textbf{Input:}}
\renewcommand{\algorithmicensure}{\textbf{Output:}}
\begin{algorithm}[t]
	\caption{Overall Algorithm for Prostrate-Beam Design}    
	 \label{xx}       
	\begin{algorithmic}[1]
	\Require  $N$, $O$, $l_f$, $l_d$, $f_c$, $h_{BS}$, $d_{\rm{max}}$, $\theta _{\rm{tilt}}$, $\rho$, ${\bf{T}}_0$.
    \vspace{1pt}
    \State  Define $ {\varphi _o} =  - \frac{\pi }{2} + \frac{{\left( {o - 1} \right)\pi }}{{O - 1}},\;\;o = 1,2,...,O.$  
    \vspace{1pt}
    \State Define ${{\bf{a}}_v}(\varphi ) \!=\! \frac{1}{{\sqrt N }}\left[ {1,{e^{j\pi \sin \varphi }},...,{e^{j\pi (N - 1)\sin \varphi }}} \right]^T$.
    \vspace{1pt}
    \State Let ${\bf{V}} = \left[ {{{\bf{a}}_v}({\varphi _1}),{{\bf{a}}_v}({\varphi _2}),...,{{\bf{a}}_v}({\varphi _O})} \right]$.
    \vspace{1pt}
    \State Calculate ${\theta _{{\rm{spread}}}} = 2\left[ {\arctan \left( {\frac{{{d_{{\rm{max}}}}}}{{{h_{{\rm{BS}}}}}}} \right) + {\theta _{{\rm{tilt}}}}} \right] - \pi$.
    \vspace{1pt}
    \State Calculate $\varphi ^+ = \frac{1}{2}\theta _{{\rm{spread}}}$ and $\varphi ^- = -\frac{1}{2}\theta _{{\rm{spread}}}$.
    \vspace{1pt}
    \State Define ${\Pi _c}$, ${\Pi _t^+}$, ${\Pi _t^-}$, and ${\Pi _s}$ based on (\ref{range}).
    \vspace{1pt}
    \State Calculate ${M_L}(\varphi _o) = l_f{f_c} \!\times\! {{l_d}^{{{\log }_{10}}\left[ {{h_{{\rm{BS}}}}/\sin (\varphi _o  + {\theta _{{\rm{tilt}}}})} \right]}}, \; \forall o$.
    \vspace{1pt}
    \State Let $ {\bf{B}} = {\bf{I}} - {{\bf{V}}^H}{\left( {{\bf{V}}{{\bf{V}}^H}} \right)^{ - 1}}{\bf{V}}$.
    \State {\bf{for}}  {$k = 0 : R$} {\bf{do}} 
    \vspace{1pt}
    \State Calculate ${{{\bf{q}}_{k}}}$, i.e., the singular vector corresponding to the largest singular value of ${{\bf{T}}_{k}}$.
    \vspace{1pt}
    \State Let ${\bf{C}}_k = {\bf{B}} + \rho \left( {{\bf{I}} - {{\bf{q}}_k}{\bf{q}}_k^H} \right)$.
    \vspace{1pt}
    \State Solve the following convex optimization problem \[\begin{aligned}
{{\bf{T}}_{k + 1}} = & \arg \mathop {\min }\limits_{\bf{T}} {\rm{Tr}}\left( {{\bf{TC}}_k} \right)\\
&{\rm{s}}.{\rm{t}}.\;{{\bf{T}}_{i,i}} = {M_L}{({\varphi _i})^2}\;{\rm{with}}\;{\varphi _i} \in {\Pi _c}\\
&\qquad  {{\bf{T}}_{j,j}} \le {M_L}{({\varphi ^+})^2}\;{\rm{with}}\;{\varphi _j} \in {\Pi _t^+},\\
&\qquad  {{\bf{T}}_{j,j}} \le {M_L}{({\varphi ^-})^2}\;{\rm{with}}\;{\varphi _j} \in {\Pi _t^-},\\
&\qquad {{\bf{T}}_{k,k}} = 0\;{\rm{with}}\;{\varphi _k} \in {\Pi _k}.
\end{aligned}\]
   \State {\bf{end for}} 
   \State Operate the eigen-decomposition ${{\bf{T}}_R} = {\bf{U\Lambda }}{{\bf{U}}^H}$.
   \vspace{1pt}
    \Ensure  ${{\bf{w}}_v} = {\left( {{\bf{V}}{{\bf{V}}^H}} \right)^{ - 1}}{\bf{VU}}\left[ {:,1} \right]\sqrt {{{\bf{\Lambda }}_{1,1}}}.$
    \end{algorithmic}
\end{algorithm}

The comprehensive steps of the optimization-based approach are delineated in Algorithm 1. To obtain the prostrate beam in the elevation direction, we first define input parameters: the number of vertical antenna elements $N$, beamspace sampling points $O$, frequency loss coefficient $l_f$, distance loss coefficient $l_d$, base station height $h_{\text{BS}}$, maximum service distance $d_{\max}$, array downtilt angle $\theta_{\text{tilt}}$, penalty parameter $\rho$, and initial optimization value $\mathbf{T}_0$. In steps 1-3, we set the direction of the sampling points, define the array steering vector, and obtain the array steering matrix. In steps 4-7, we calculate the beam spread required for precise coverage of the service region through Proposition 1, and obtain the coverage range, transition range, stop range of the prostrate beam, and the beam gain (in magnitude) required in different directions. Steps 8-12 implement our customized SCA algorithm for iterative prostrate beam optimization. At iteration $R$, we perform eigen-decomposition on $\mathbf{T}_R$ to extract the dominant eigenvector and obtain our expected prostrate beam pattern $\mathbf{w}_v$. In particular, the total iteration count $R$ can be established empirically or ascertained by a convergence criterion, i.e., stopping iteration when the increase in the objective value falls below $\epsilon$. Let $g_k$ denote the objective value of $({\rm{P6}})_k$ in the $k$-th iteration. One can see that
\begin{align}\label{conver}
{g_k} &=  {\rm{Tr}}\left( {{\bf{T}_{k+1}}{{\bf{C}}_k}} \right)\\
&={\rm{ Tr}}\left( {{{\bf{T}}_{k + 1}}{\bf{B}}} \right) + \rho {\rm{Tr}}\left( {{{\bf{T}}_{k + 1}}} \right) - \rho {h_k}\left( {{{\bf{T}}_{k + 1}};{{\bf{T}}_k}} \right). \notag\\
&\mathop  \ge \limits^{(a)} {\rm{Tr}}\left( {{{\bf{T}}_{k + 1}}{\bf{B}}} \right) + \rho {\rm{Tr}}\left( {{{\bf{T}}_{k + 1}}} \right) - \rho {h_k}\left( {{{\bf{T}}_{k + 1}};{{\bf{T}}_{k + 1}}} \right)\notag\\
&\mathop  \ge \limits^{(b)} \mathop {\min }\limits_{\bf{T}} {\rm{Tr}}\left( {{\bf{TB}}} \right) + \rho {\rm{Tr}}\left( {\bf{T}} \right) - \rho {h_k}\left( {{\bf{T}};{{\bf{T}}_{k + 1}}} \right) =g_{k+1}, \notag
\end{align}
where the inequality $(a)$ is due to the fact ${h_k}\left( {{\bf{T}};{{\bf{T}}_k}} \right) \le {\left\| {\bf{T}} \right\|_2}$, as ${h_k}\left( {{\bf{T}};{{\bf{T}}_k}} \right)$ is the first-order Taylor
approximation of the convex function ${\left\| {\bf{T}} \right\|_2}$. Thus, we have 
\begin{equation}
{h_k}\left( {{{\bf{T}}_{k + 1}};{{\bf{T}}_k}} \right) \le {\left\| {{{\bf{T}}_{k + 1}}} \right\|_2} = {h_k}\left( {{{\bf{T}}_{k + 1}};{{\bf{T}}_{k + 1}}} \right).    
\end{equation}
The inequality $(b)$ is due to the fact that the objective of a feasible solution is no less than that of the optimal solution. Consequently, according to (\ref{conver}), each iteration of the optimization is non-increasing and the convergence is guaranteed. The CVX solver is utilized to manage the computational complexity, which is $\mathcal{O}({RN^{5.5}+O^3})$. It is important to note that the computational complexity exerts minimal influence on the design of the prostrate beam, as the codewords are predefined in an offline codebook, with only online selection performed in centimeter-wave G-MIMO systems.

\section{Numerical Results}\label{secnu}

\begin{figure}[t]
\centering
\includegraphics[width=3.3in]{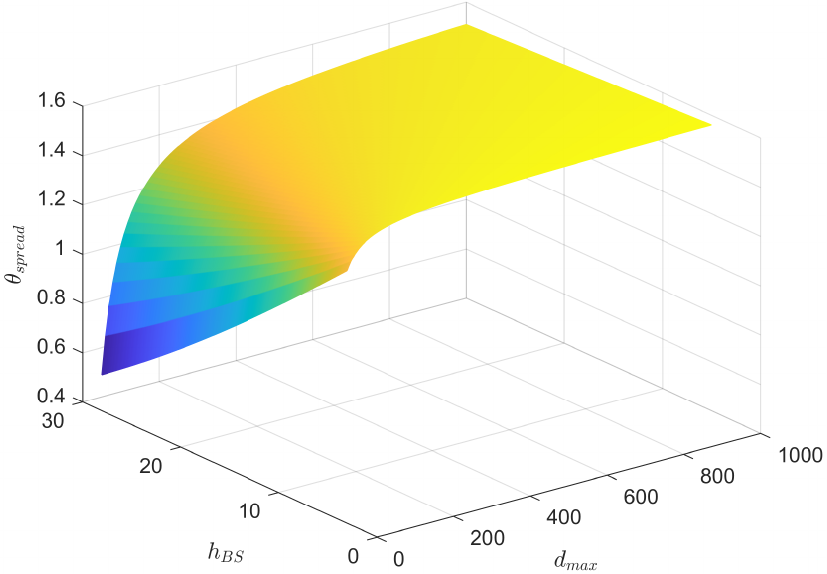}
\caption{The relationship between the base station height, beam spread, and maximum service distance.}\label{proposition}
\vspace{-8pt}
\end{figure}

\begin{figure}[t]
\centering
\includegraphics[width=3.5in]{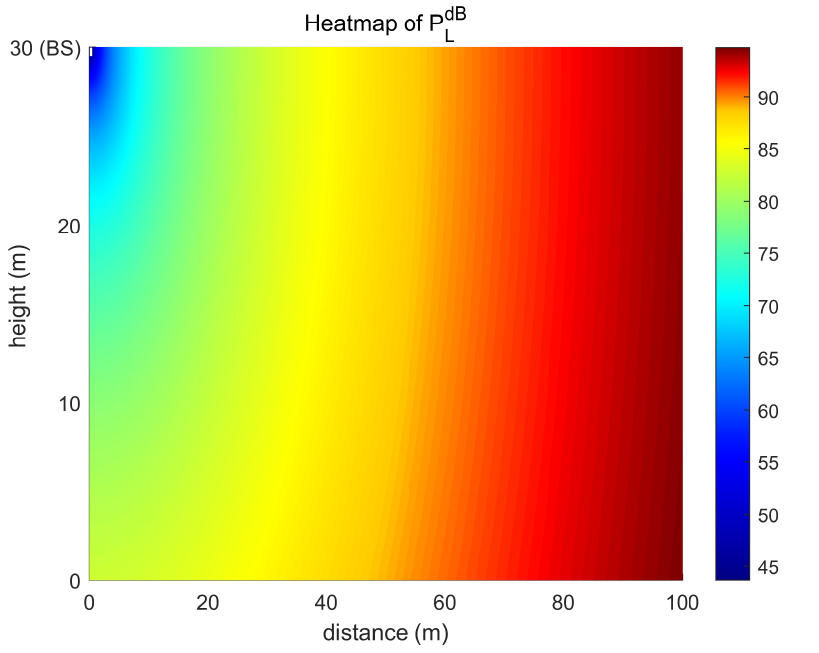}
\caption{Propagation loss vs. user's distance and height.}\label{heatmap}
\vspace{-8pt}
\end{figure}

In this section, we present numerical results to demonstrate the performance of the traditional wide beam and the proposed prostrate beam, enabled by beam shaping technology. The traditional wide beam is constructed following the approach of S-SARV\cite{trad4}. We begin by defining the symmetrical array response vector, which differs from the conventional array response vector by assuming the signal phase at the center of a uniform linear array (ULA) is zero, rather than at the first antenna. This vector is expressed as
\begin{align}
 \hat{\mathbf{a}}_{N}(\varphi) = \frac{1}{\sqrt{N}} &\left[ e^{j\pi\left(-\frac{N-1}{2}\right)\sin\varphi}, \ldots\right. \\
&\qquad \left. \ldots, e^{j\pi n\sin\varphi}, \ldots, e^{j\pi\left(\frac{N-1}{2}\right)\sin\varphi} \right]^{T},\notag   
\end{align} 
where \(n\) takes values from \(-\frac{(N-1)}{2}+1\) to \(\frac{(N-1)}{2}-1\). A set of \(O\) angles (with \( \geq 2N\)) is then defined as
\begin{equation}
\begin{aligned}
{\varphi _o} = \arcsin &\Bigg[ {\sin {\theta ^ - } - \frac{1}{N}}\\
&  \qquad+ \frac{{(o - 1)\left( {\sin {\theta ^ + } - \sin {\theta ^ - } + \frac{2}{N}} \right)}}{{O - 1}} \Bigg]
\end{aligned}
\end{equation}
for \(o = 1, 2, \ldots, O\). Here, \(\theta^{-}\) and \(\theta^{+}\) represent the lower and upper angular bounds of the intended coverage region, and \(N\) is the number of transmit antennas. Based on this setup, a wide beamforming weight vector is constructed by directly summing the symmetrical array response vectors corresponding to these sampled angles, i.e.,
\begin{equation}
\mathbf{w}_{\text{S-SARV}} = \beta \sum_{o=1}^{O} \hat{\mathbf{a}}_{N}(\varphi_{o}),    
\end{equation}
where the normalization factor \(\beta\) is given by
\begin{equation}
\beta = \frac{\sqrt{N\left(\sin\theta^{+} - \sin\theta^{-} + \frac{2}{N}\right)}}{\sqrt{2} \cdot \left\| \sum_{o=1}^{O} \hat{\mathbf{a}}_{N}(\varphi_{o}) \right\| }.   
\end{equation}

\begin{figure}[t]
\centering
\includegraphics[width=3in]{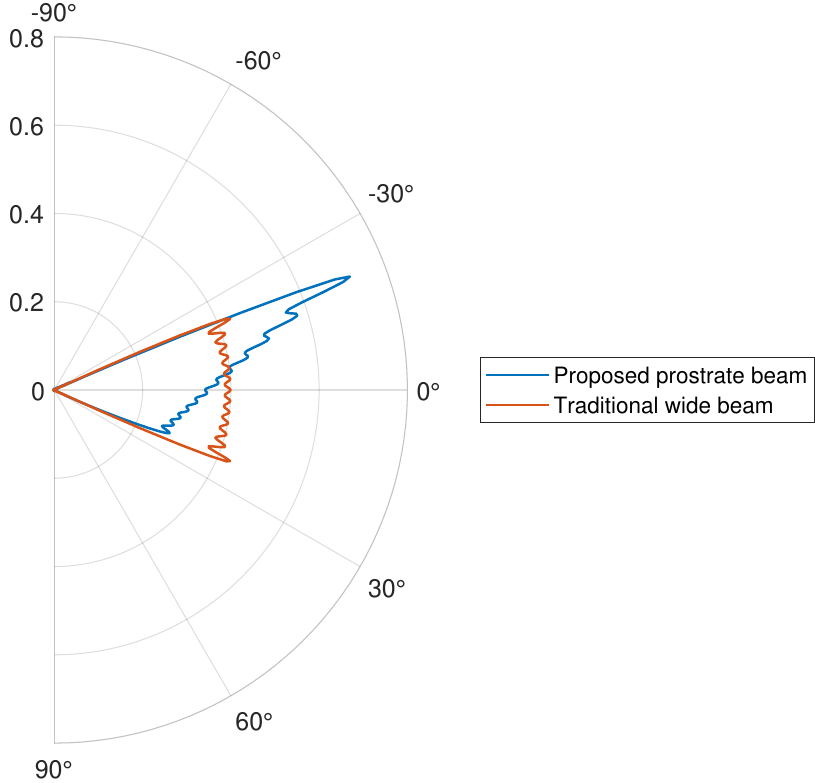}
\caption{The beam pattern of prostrate beam and traditional wide beam.}
\label{beampattern}
\vspace{-8pt}
\end{figure}

\begin{figure}[t]
\centering
\includegraphics[width=3.1in]{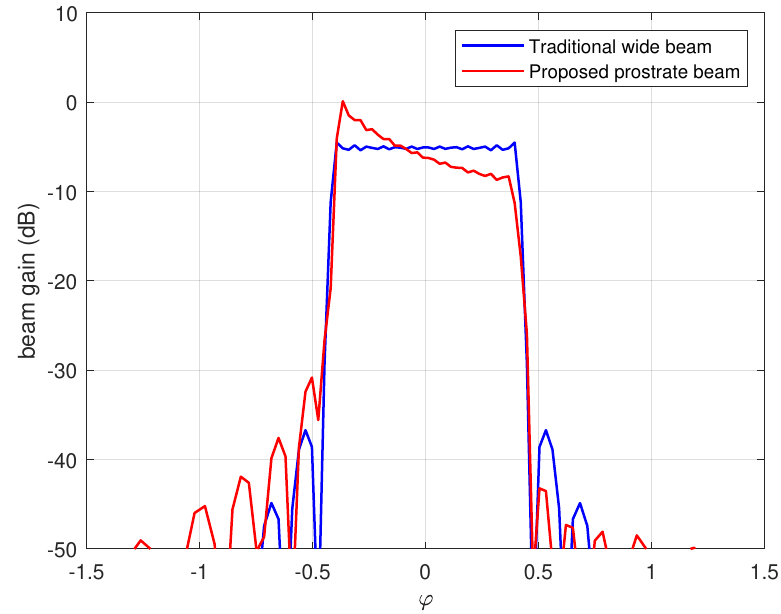}
\caption{Beam gain versus different elevation angle.}\label{beamcom}
\vspace{-8pt}
\end{figure}

\begin{figure}[t]
\centering
\includegraphics[width=3.2in]{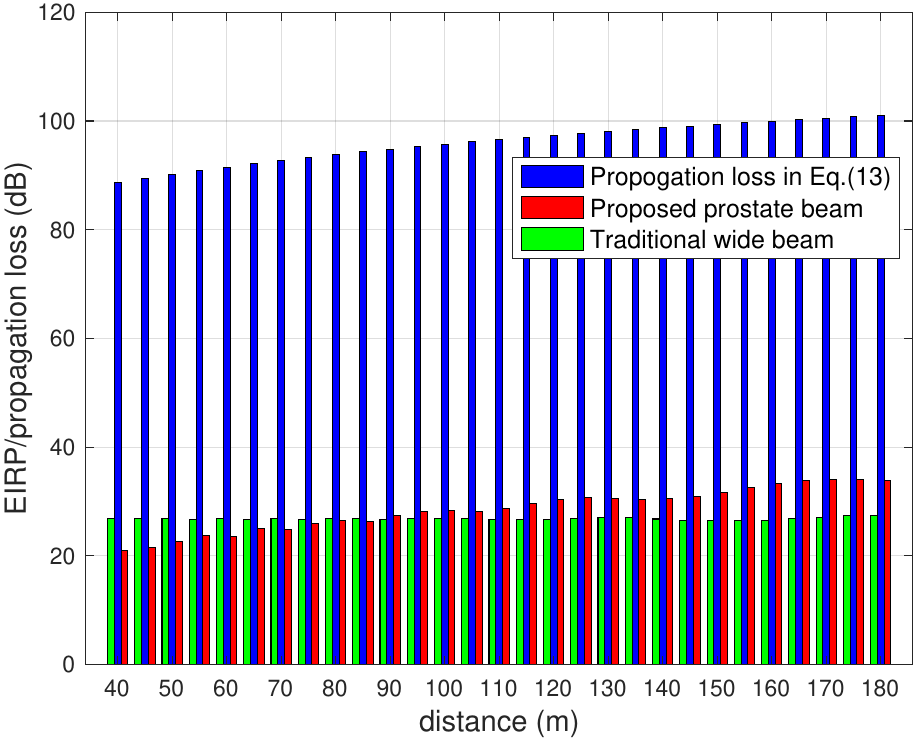}
\caption{Propagation loss and EIRP across different distances.}
\label{barchart}
\vspace{-8pt}
\end{figure}

\begin{figure}[t]
\centering
\includegraphics[width=3.2in]{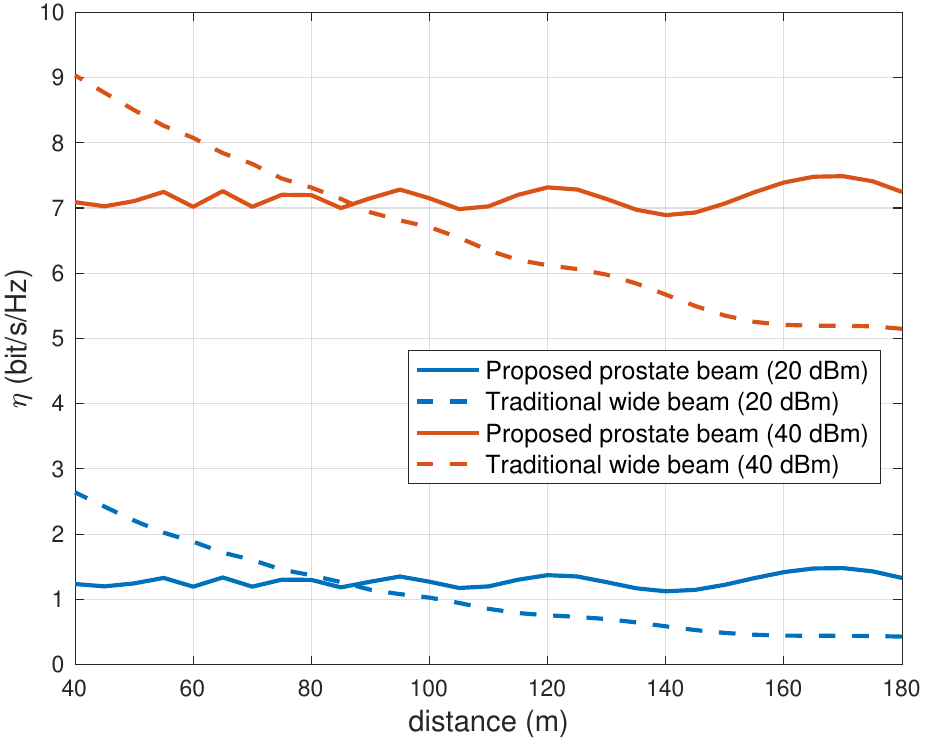}
\caption{Spectral efficiency across different distances.}
\label{efficiency}
\vspace{-8pt}
\end{figure}

To the best of our knowledge, among existing wide-beam design methods, the S-SARV method offers the best overall trade-off between pattern performance and computational complexity, making it a practical and efficient choice for wide-beam synthesis in large-scale antenna systems\cite{ning2023beamforming}. \emph{As this work is not a comparative study of wide-beam designs, we employ this single benchmark as a reference. We note that other existing techniques produce beams with similar shapes, making their inclusion unnecessary for our core analysis.} Unless otherwise specified, the simulation parameters are set as follows: the number of antenna elements is \(N=64\), the carrier frequency is \(f_c=13\ \text{GHz}\), the base station height is \(h_\text{BS}=30\ \text{m}\), the tilt angle is \(\theta_\text{tilt}=\pi/4\), the number of sampling points is \(O=2N=128\), the penalty factor is \(\rho=100\), and the number of SCA iterations is $10$.
\begin{figure*}[t] 
\centering
\includegraphics[width=0.94\textwidth]{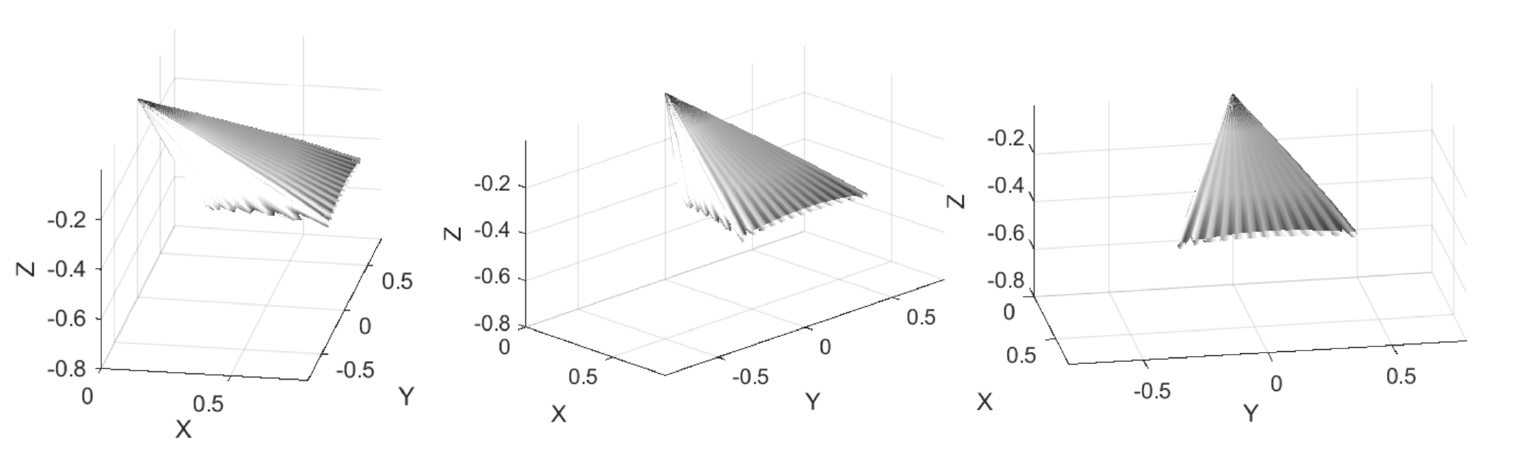} 
\caption{Illustration of 3D beam shaping via the Kronecker product of a vertical prostrate beam and a horizontal wide beam.}
\label{3dplot}
\includegraphics[width=6.8in]{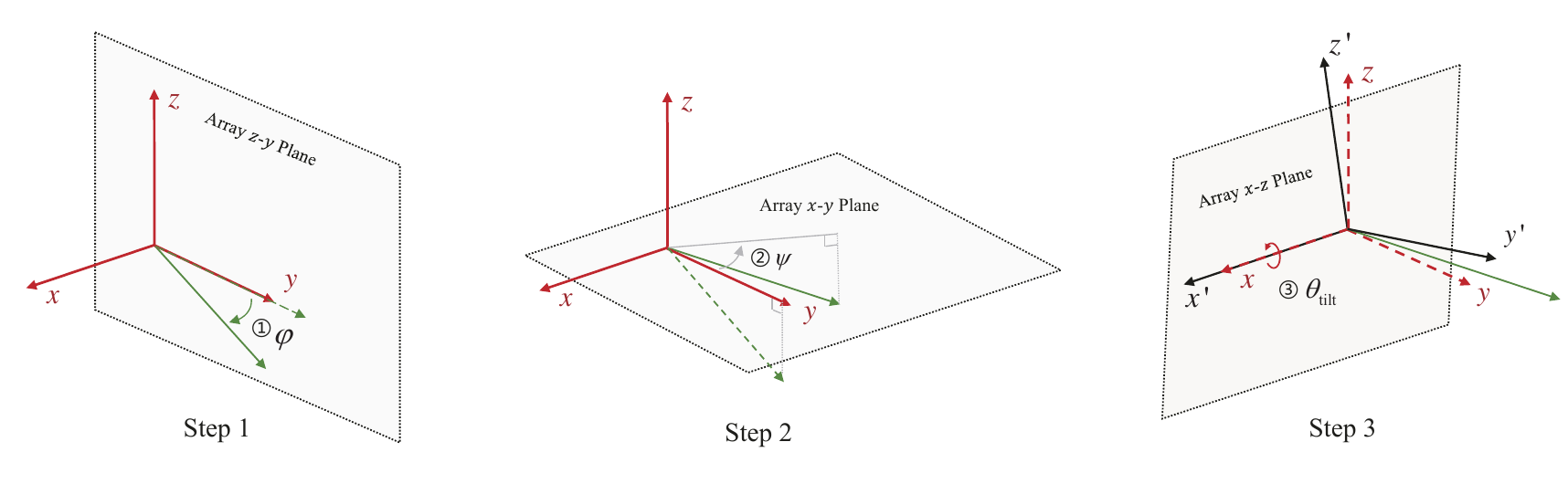}
\caption{Three steps to calculate the direction of a target point in the global coordinate system.}\label{coordinate}
\vspace{-8pt}
\end{figure*}

Fig. \ref{proposition} plots the relationship between the base station height \(h_{\text{BS}}\), beam spread \(\theta_{\text{spread}}\), and maximum service distance \(d_{\text{max}}\). The simulation results show that, under a fixed antenna tilt of \(\theta_{\text{tilt}}=\pi/4\), the coverage distance increases with both the base station height and the beam spread. For a given height, \(d_{{max}}\) grows rapidly as \(\theta_{\text{spread}}\) increases, particularly within the range of \(0^\circ\) to \(60^\circ\), after which the increase becomes more gradual. Similarly, for a fixed beam spread, higher BS installations extend service coverage almost linearly, with a slope determined by \(\theta_{\text{spread}}\). The largest \(d_{\text{max}}\) values (approaching the limit of $1000$ m) are achieved with tall arrays and wide beam spreads (e.g., \(h_{\text{BS}}=30\) m, \(\theta_{\text{spread}} \approx 86^\circ\)). These trends highlight the critical role of adjusting both physical base station height and angular beam spread to meet targeted coverage requirements in practical deployments.

Fig. \ref{heatmap} presents a heatmap of the propagation loss versus the user's horizontal distance and height, with the base station deployed at a height of $30$ m and the user's distance ranging from $0$ to $100$ m.  It can be observed that at the same height, the propagation loss increases with distance. Similarly, at a fixed distance, the loss gradually increases as the height decreases. Notably, the rate of increase in loss is most rapid at distance = $0$ m and slows down as the distance increases. At ground level, the radiation loss varies significantly across different distances—ranging from $84$ dB at 0 m to $95$ dB at $100$ m. Interestingly, the propagation loss from the base station in the vertical direction is remarkably high. At zero horizontal distance and ground level, the path loss is about $84$ dB, while with the base station fixed at $30$ m height, moving the user from $0$ m to $30$ m in distance results in only about $3$ dB additional loss.

Fig. \ref{beampattern} illustrates the beam patterns of the proposed prostrate beam and the traditional wide beam in polar coordinates. We simulated the case with a beam spread of $\pi/4$ and normalized the transmit power by letting $P=1$. As shown, the side lobes are nearly invisible in this beam pattern representation because both beams effectively suppress sidelobes, allowing us to focus primarily on the differences in the main lobe shape. Furthermore, while the beam gain of the prostrate beam decreases at a progressively slower rate with increasing angle within the main lobe in the dB plot, its gain in the beam pattern exhibits a linear decline. Given that the array's tilt angle is set to $45$ degrees, if we rotate the beam pattern clockwise by $45$ degrees, we can observe that the beam energy advances forward while closely skimming the ground. This effectively demonstrates that the intended ``prostrate" characteristic has been successfully achieved.

Fig. \ref{beamcom} illustrates the beam gain of the traditional wide beam and the proposed prostrate beam. This visualization is equivalent to converting the linear beam gain into dB-scale values, allowing us to visually and intuitively evaluate the performance of the compared beams. It can be observed that the coverage range for both beams, where beam gain exceeds $-20$ dB, is essentially equal to the beam spread. Notably, the traditional wide beam maintains a nearly uniform gain within its coverage area, remaining stable at approximately $-5$ dB, while the proposed prostrate beam exhibits a gradual decrease in beam gain as the angle increases across the coverage range. Since $\varphi = 0$ corresponds to the boresight direction of the array, $\varphi < 0$ points toward the sky, and $\varphi > 0$ points toward the ground, the proposed prostrate beam can provide higher beam gain in farther directions. Moreover, the proposed beam shows no significant dips within the coverage range, indicating no coverage blind spots.

Fig. \ref{barchart} presents the propagation loss and the effective isotropic radiated power (EIRP) of different schemes across various distances. Here, EIRP is defined as the transmit power multiplied by the beam gain for the target point at a certain distance. After fixing the base station height, array tilt angle, and beam spread, the coverage range of the beam corresponds to ground distances from $40$ m to $180$ m, over which the propagation loss is also plotted. During the design of the prostrate beam, we calculated the minimum total EIRP required to achieve gains matching the sum of squares of the $M_L$ values in \ref{ml} across all angles within the coverage range, and we set the total EIRP slightly higher than this minimum value. To ensure a fair comparison, the total EIRP of the traditional wide beam and the proposed prostrate beam is identical, meaning the total area of the green bars (representing one scheme) and the red bars (representing the other) is the same. As can be seen, the EIRP of the proposed prostrate beam generally increases in tandem with the propagation loss. Assuming a demodulation threshold of $-70$ dB, the prostrate beam ensures that the EIRP can overcome the propagation loss across the entire coverage range. In contrast, the traditional wide beam can only cover distances from $40$ m to $100$ m, which is less than $50\%$ of the coverage capability.

Fig.~\ref{efficiency} illustrates the spectral efficiency as a function of the communication distance. In our analysis, the receive array gain is set to $G_r = 4$~dB, and the received power is determined by the standard link budget equation, i.e.,
\begin{equation}
P_r = \mathrm{EIRP} + G_r - P_{L}^{\mathrm{dB}},
\end{equation}
where $P_{L}^{\mathrm{dB}}$ denotes the total path loss in dB. The noise power at the receiver input is primarily governed by the receiver noise figure (NF) and the system bandwidth $B$. It is calculated as
\begin{equation}
N^{\mathrm{dBm}} = -174 + 10\log_{10}(B) + NF,
\end{equation}
with $B$ in Hz and $NF$ in dB. Accordingly, the signal-to-noise ratio (SNR) is obtained via ${\rm{SNR}} = P_r - N^{\mathrm{dBm}}$ (in dB). For the numerical evaluation, we adopt a bandwidth of $400$~MHz and a noise figure of $3$~dB. The corresponding spectral efficiency is then derived from the Shannon capacity formula, i.e., $\eta = \log_2(1 + 10^{\rm{SNR}/10})$ in bit/s/Hz, which serves as the upper bound for the achievable data rate per unit bandwidth. We compare the spectral efficiency under transmit powers of $20$ dBm and $40$ dBm, respectively. It can be observed that our proposed prostrate beam enables cell-edge users to achieve significantly higher link performance, and the gain becomes more pronounced as the distance increases.




Fig. \ref{3dplot} illustrates the 3D prostrate beam constructed using our proposed beam-shaping method. As noted in Remark 3, the method presented in this work can be applied not only to ULAs for generating customized 2D beams, but also to UPAs for obtaining customized 3D beams. To achieve this, one only needs to design the beamforming vectors for the horizontal and vertical ULAs separately, and then take their Kronecker product to form a beamforming matrix for exciting the UPA antenna array. For the 3D beam in Fig. \ref{3dplot}, we first obtain a prostrate beam with a beam spread of $\pi/4$ using the proposed method. This beam is then multiplied via the Hadamard product with an array steering vector oriented at $\pi/4$ to generate a downward-projected prostrate beam. Finally, the 3D prostrate beam is produced by taking the Kronecker product with a conventional wide beam. Three different viewing angles are captured by rotating the beam horizontally to demonstrate its 3D radiation pattern. It can be observed that the beam gain is uniformly distributed along the y-axis, while increasing along the x-axis as the distance grows. The bottom of the beam is nearly parallel to the ground, confirming the effectiveness of our design.

\section{Conclusion}
In this paper, we addressed the design of broadcasting beams in centimeter-wave G-MIMO communication systems and proposed a beam shaping technique capable of achieving arbitrary desired coverage patterns, provided that the array possessed sufficient spatial degrees of freedom. This beam shaping approach enabled scenario-specific broadcast beam design, particularly allowing customization in the horizontal direction based on user distribution or obstacle layouts.

In the vertical direction, the ground represented a consistent environmental factor in most scenarios. Leveraging beam shaping, we designed a \emph{prostrate beam} whose radiation energy closely followed the ground surface. This design enhanced the coverage range compared to traditional wide beams. However, it required that the total beam energy exceed a theoretical full-coverage threshold; otherwise, the prostrate beam could result in zero coverage, underperforming even traditional beams. In such cases, no beam would be able to achieve full coverage across the intended service area.

Since broadcast beams are typically implemented through codebook-based scanning and are often optimized offline, there is no stringent requirement on computational complexity. Therefore, the beam shaping method proposed in this paper is widely applicable to most centimeter-wave G-MIMO communication systems that rely on broadcast mechanisms. Future research can focus on beam shaping technologies with additional degrees of freedom or constraints. Promising directions include the joint optimization of antenna positions and weights for movable antenna systems, as well as beam shaping with constant modulus constraints tailored for analog array architectures.
\begin{appendices}    

\section{Proof of Proposition 1}
Based on the geometric relationships presented in Fig. \ref{model}, it is easy to obtain that
\begin{equation}
{d_{{\rm{max}}}} = {h_{{\rm{BS}}}} \cdot \tan \left( {\frac{{{\theta _{{\rm{spread}}}} + \pi }}{2} - {\theta _{{\rm{tilt}}}}} \right).    
\end{equation}
Then, the results in Proposition 1 can be derived through trigonometric functions and equation transformation.

\section{Proof of Proposition 2}

Given the array's coordinate system is $x-y-z$, we further consider a global coordinate system $x'-y'-z'$, where the direction of $x'$ aligns with $x$, and the direction of $z'$ aligns with the BS pole. In the global coordinate system, the direction of the BS pole towards the ground can be represented as 
\begin{equation}
 {\bf{v}}_{\rm{pole}}=[0,0,-1]^T. 
\end{equation}
Next, we calculate the target direction in the global coordinate system via three steps as shown in Fig. {\ref{coordinate}}. In the first two steps, we derive the target direction in the array's coordinate system. In the last step, we derive the target direction in the global coordinate system. Specifically, let consider a reference direction starting from the $y$-axis, i.e., 
\begin{equation}
 {\bf{v}}_0=[0,1,0]^T. 
\end{equation}
By rotating the reference direction clockwise by $\varphi$ around the $x$-axis, the new direction of the reference becomes
\begin{equation}
{{\bf{v}}_1} = \left[ {\begin{array}{*{20}{c}}
1&0&0\\
0&{\cos \varphi }&{\sin \varphi }\\
0&{ - \sin \varphi }&{\cos \varphi }
\end{array}} \right]{{\bf{v}}_0} = \left[ {\begin{array}{*{20}{c}}
0\\
{\cos \varphi }\\
{ - \sin \varphi }
\end{array}} \right].
\end{equation}
Then, rotating the reference direction counterclockwise by $\psi$ around the $z$-axis, the new reference direction becomes
\begin{equation}
{{\bf{v}}_2} \!=\! \left[ {\begin{array}{*{20}{c}}
{\cos \psi }&{ - \sin \psi }&0\\
{\sin \psi }&{\cos \psi }&0\\
0&0&1
\end{array}} \right]{{\bf{v}}_1} \!=\! \left[ {\begin{array}{*{20}{c}}
{ - \sin \psi \cos \varphi }\\
{\cos \psi \cos \varphi }\\
{ - \sin \varphi }
\end{array}} \right].
\end{equation}
Finally, by rotating the array's coordinate system counterclockwise by $\theta _{\rm{tilt}}$ around the $x$-axis, we obtain the global coordinate system. The reference direction in the global coordinate system is given by
\allowdisplaybreaks
\begin{align}
{{\bf{v}}_3} &= \left[ {\begin{array}{*{20}{c}}
1&0&0\\
0&{\cos {\theta _{{\rm{tilt}}}}}&{\sin {\theta _{{\rm{tilt}}}}}\\
0&{ - \sin {\theta _{{\rm{tilt}}}}}&{\cos {\theta _{{\rm{tilt}}}}}
\end{array}} \right]{{\bf{v}}_2}  \notag \\
 &= \left[ {\begin{array}{*{20}{c}}
{ - \sin \psi \cos \varphi }\\
{\cos \psi \cos \varphi \cos {\theta _{{\rm{tilt}}}} - \sin \varphi \sin {\theta _{{\rm{tilt}}}}}\\
{ - \cos \psi \cos \varphi \sin {\theta _{{\rm{tilt}}}} - \sin \varphi \cos {\theta _{{\rm{tilt}}}}}
\end{array}} \right].
\end{align}
Thus, the angle between the BS pole and the target direction can be calculated using 
\begin{equation}
\begin{aligned}
{\theta _{{\rm{target}}}} &= \arccos \left( {{\bf{v}}_3^T{{\bf{v}}_{{\rm{pole}}}}} \right)\\
 &= \arccos \left( {\cos \psi \cos \varphi \sin {\theta _{{\rm{tilt}}}} + \sin \varphi \cos {\theta _{{\rm{tilt}}}}} \right)
\end{aligned}   
\end{equation}
By leveraging the fact that
\begin{equation}
\cos\left( {{\theta _{{\rm{target}}}}} \right) = \frac{{{h_{{\rm{BS}}}}}}{{{d_{3D}}}},
\end{equation}
we can obtain the results given in Proposition 2.

\section{Proof of the relation in \eqref{SAP}}
The integration of the radiation power in the beam domain can be computed by
\begin{equation}
\begin{aligned}
&{\int_{ - 1}^1 {\int_{ - 1}^1 {\left| {\hat r({\psi _h},{\varphi _v})} \right|} ^2 }}d{\psi _h}d{\varphi _v}\\
 &= {\int_{ - 1}^1 {\int_{ - 1}^1 {\left| {\sum\nolimits_{n,m} {{{\left[ {{{{\bf{\hat A}}}^*}({\psi _h},{\varphi _v}) \circ {\bf{W}}} \right]}_{n,m}}} } \right|} } ^2}d{\psi _h}d{\varphi _v}\\
 &= {\int_{ - 1}^1 {\int_{ - 1}^1 {\left| {\sum\limits_{m = 0}^{M - 1} {{e^{jm\pi{\psi _h}}}} \sum\limits_{n = 0}^{N - 1} {{e^{jn\pi{\varphi _v}}}} {{\bf{W}}_{n,m}}} \right|} } ^2}\frac{{d{\psi _h}d{\varphi _v}}}{{MN}}.
\end{aligned}   
\end{equation}
Define an auxiliary function as 
\begin{equation}
f\left( {{\varphi _v},m} \right) = \sum\limits_{n = 0}^{N - 1} {{e^{jn\pi{\varphi _v}}}} {{\bf{W}}_{n,m}}.
\end{equation}
Then, the integration can be further calculated as
\begin{equation}\label{ap1}
\begin{aligned}
&{\int_{ - 1}^1 {\int_{ - 1}^1 {\left| {\sum\limits_{m = 0}^{M - 1} {{e^{jm\pi{\psi _h}}}} \sum\limits_{n = 0}^{N - 1} {{e^{jn\pi{\varphi _v}}}} {{\bf{W}}_{n,m}}} \right|} } ^2}\frac{{d{\psi _h}d{\varphi _v}}}{{MN}}\\
 &= \int_{ - 1}^1 {\left[ {{{\int_{ - 1}^1 {\left| {\sum\limits_{m = 0}^{M - 1} {{e^{jm\pi{\psi _h}}}} f\left( {{\varphi _v},m} \right)} \right|} }^2}d{\psi _h}} \right]} \frac{{d{\varphi _v}}}{{MN}}.
\end{aligned}   
\end{equation}
The expression in the square bracket can be expanded as 
\begin{align}\label{ap2}
&{\int_{ - 1}^1 {\left| {\sum\limits_{m = 0}^{M - 1} {{e^{j\pi m{\psi _h}}}} f\left( {{\varphi _v},m} \right)} \right|} ^2}d{\psi _h} \notag\\
 &= \int_{ - 1}^1 {\sum\limits_{p = 0}^{M - 1} {\sum\limits_{q = 0}^{M - 1} {{f^*}\left( {{\varphi _v},p} \right)f\left( {{\varphi _v},q} \right)} } } {e^{j\pi (q - p){\psi _h}}}d{\psi _h} \notag\\
 &= \sum\limits_{m = 0}^{M - 1} {{f^*}\left( {{\varphi _v},m} \right)f\left( {{\varphi _v},m} \right)} \int_{ - 1}^1 {1d{\psi _h}}\notag \\
 & \quad + \underbrace {\sum\limits_{p = 0}^{M - 1} {\sum\limits_{q \ne p}^{M - 1} {{f^*}\left( {{\varphi _v},p} \right)f\left( {{\varphi _v},q} \right)d{\psi _h}} } \int_{ - 1}^1 {{e^{j\pi (q - p){\psi _h}}}d{\psi _h}} }_{ = 0,\;{\rm{ since }}\int_{ - 1}^1 {{e^{j\pi (q - p){\psi _h}}}d{\psi _h}}  = 0}\notag\\
& = 2\sum\limits_{m = 0}^{M - 1} {{{\left| {f\left( {{\varphi _v},m} \right)} \right|}^2}}. 
\end{align}
Substituting (\ref{ap2}) into (\ref{ap1}), we have 
\begin{equation}
\begin{aligned}
&\int_{ - 1}^1 {\left[ {{{\int_{ - 1}^1 {\left| {\sum\limits_{m = 0}^{M - 1} {{e^{j\pi m{\psi _h}}}} f\left( {{\varphi _v},m} \right)} \right|} }^2}d{\psi _h}} \right]} \frac{{d{\varphi _v}}}{{MN}}\\
& = 2\int_{ - 1}^1 {\sum\limits_{m = 0}^{M - 1} {{{\left| {f\left( {{\varphi _v},m} \right)} \right|}^2}} } \frac{{d{\varphi _v}}}{{MN}}\\
& = 2\sum\limits_{m = 0}^{M - 1} {\int_{ - 1}^1 {{{\left| {\sum\limits_{n = 0}^{N - 1} {{e^{jn\pi{\varphi _v}}}} {{\bf{W}}_{n,m}}} \right|}^2}} \frac{{d{\varphi _v}}}{{MN}}} \\
&\mathop  = \limits^{(a)} \frac{4}{{MN}}{\sum\limits_{n = 0}^{N - 1} {\sum\limits_{m = 0}^{M - 1} {\left| {{{\bf{W}}_{n,m}}} \right|^2}}},
\end{aligned}
\end{equation}
where Eq. (a) comes from the similar derivation provided in (\ref{ap2}), which completes the proof.

\end{appendices}

\bibliographystyle{IEEEtran}
\bibliography{IEEEabrv, mybib, jp_refs}

\begin{IEEEbiography}[{\includegraphics[width=1in,clip]{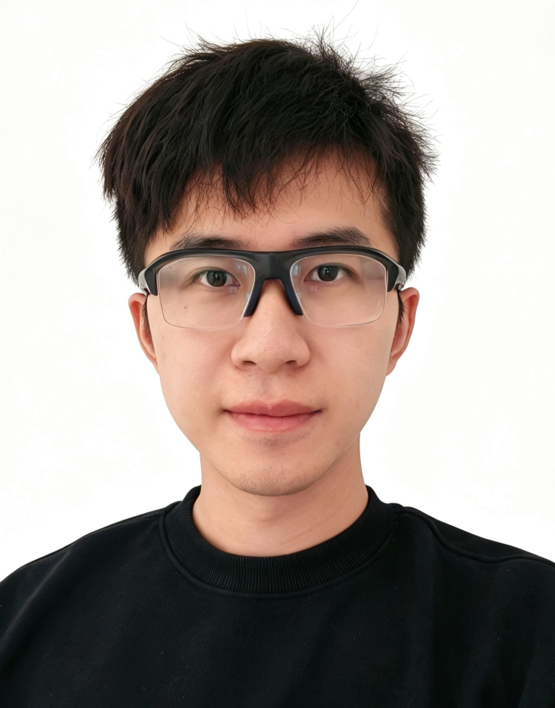}}]{Boyu Ning} (Member, IEEE) received the B.S. degree from the Yingcai Honors College, University of Electronic Science and Technology of China (UESTC), Chengdu, China, in 2018, and the Ph.D. degree from the National Key Laboratory of Wireless Communications, UESTC, in 2023. He was a Visiting Scholar with the Department of Electrical and Computer Engineering, National University of Singapore, from 2022 to 2023. He was a Principal Engineer with the Wireless Networks System Research Department II, Huawei Technologies Co., Ltd., Chengdu, China, from 2023 to 2026. He is currently a Postdoctoral Fellow with the Department of Communication Systems, KTH Royal Institute of Technology, Stockholm, Sweden. He has experience in the communication industry, including the evolution of 5G base station architectures, the baseband algorithms design for 5G NR, and the standardization for 3GPP Release 19/20. His research interests include massive MIMO, deep learning, Bayesian estimation, compressed sensing, and tensor decomposition.

Dr. Ning has published over 100 top-tier IEEE journal and conference papers, including 5 ESI hot papers and 13 ESI highly cited papers. He has secured over 40 national patents related to 3GPP standardization. He was a recipient of the National Scholarship for Ph.D. Students in 2022, Huawei's Top Minds in 2023, the Outstanding Doctoral Dissertation Award from the China Education Society of Electronics in 2024, and the MSCA Postdoctoral Fellowship from the European Commission in 2026. He has been listed among the World's Top 2\% Scientists by Stanford University since 2025. He has served as an industry panel moderator and a workshop chair for various global conferences, including IEEE Globecom, ICC, VTC, ICCC, and WCNC.
\end{IEEEbiography}

\begin{IEEEbiography}[{\includegraphics[width=1in,clip]{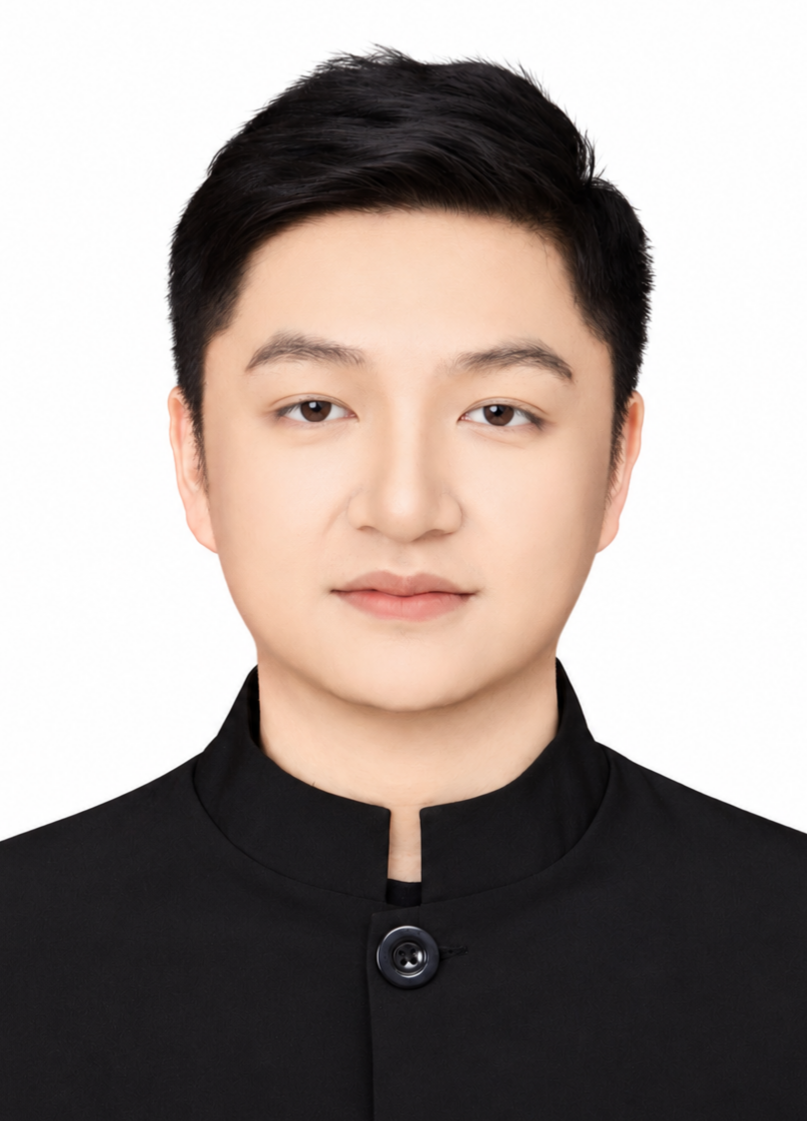}}]{Yike Xie} (Student Member, IEEE)  received the B.S. degree in automation engineering from the School of Automation Engineering, University of Electronic Science and Technology of China (UESTC), Chengdu, China. He is currently pursuing the Ph.D. degree in information and communication engineering with the National Key Laboratory of Wireless Communications, UESTC. He is also a visiting student with the Department of Electronical and Computer Engineering, National University of Singapore. His research interests include terahertz communications, movable antennas and massive MIMO. He was the co-recipients of the National Scholarship from MOE of the People's Republic of China in 2019 and the IEEE ComSoc SPCC Technical Committee Student Challenge and Video Contest Award in 2024.
\end{IEEEbiography}

\begin{IEEEbiography}[{\includegraphics[width=1in,clip]{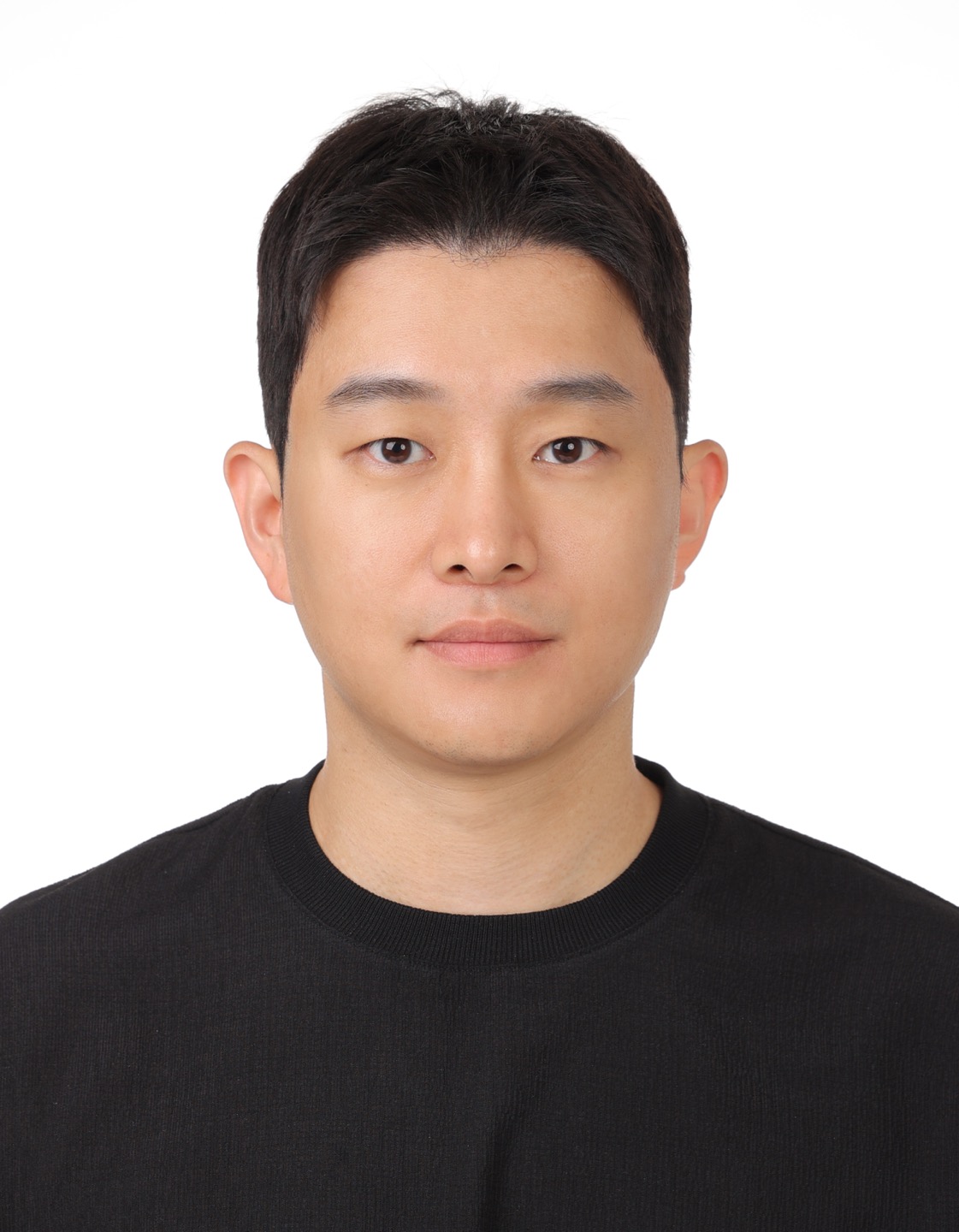}}] {Jaebum Park} (Student Member, IEEE) received the B.S. degree from Yonsei University, Seoul, South Korea, in 2015 and the M.S. degree from the Korea Advanced Institute of Science and Technology (KAIST), Daejeon, South Korea, in 2017, both in Electrical and Electronic Engineering. From 2017 to 2024, he was with the 6G Research Team at Samsung Research, Seoul, South Korea. He is currently pursuing the Ph.D. degree with the Department of Electrical and Computer Engineering at the University of California, San Diego. His research interests include advanced MIMO transceiver algorithms, energy-efficient MIMO techniques for next-generation RAN architectures, and wireless communications in upper-midband frequency bands.
\end{IEEEbiography}

\begin{IEEEbiography}[{\includegraphics[width=1in,clip]{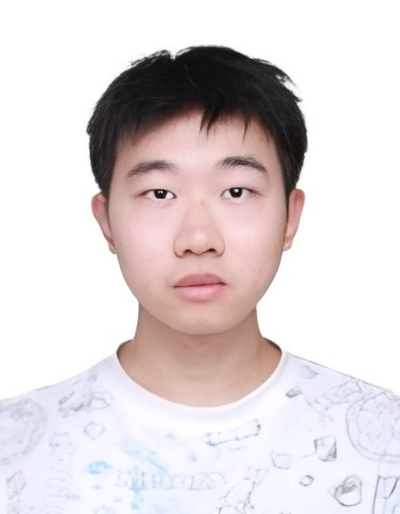}}]{Songjie Yang} received the B.E. degree in Communication Engineering from Shanghai University in 2020, and the Ph.D. degree with the National Key Laboratory of Wireless Communications, University of Electronic Science and Technology of China, in 2026. He was a visiting Ph.D. student under a joint training program at Singapore University of Technology and Design and Nanyang Technological University, Singapore. He is currently a Postdoctoral Fellow with the Beijing National Research Center for Information Science and Technology (BNRist), Tsinghua University. His research interests include reconfigurable MIMO, near-field extremely large-scale antenna array communications, integrated sensing and communications, and sparse signal processing. 

Dr. Yang has published over 40 top-tier IEEE journal and conference papers, including 7 ESI highly cited papers. He has been listed among the World's Top 2\% Scientists by Stanford University since 2026.  He served as a TPC Member for the IEEE ICC, WCNC, and IEEE Globecom. 
\end{IEEEbiography}

\begin{IEEEbiography}[{\includegraphics[width=1in,clip]{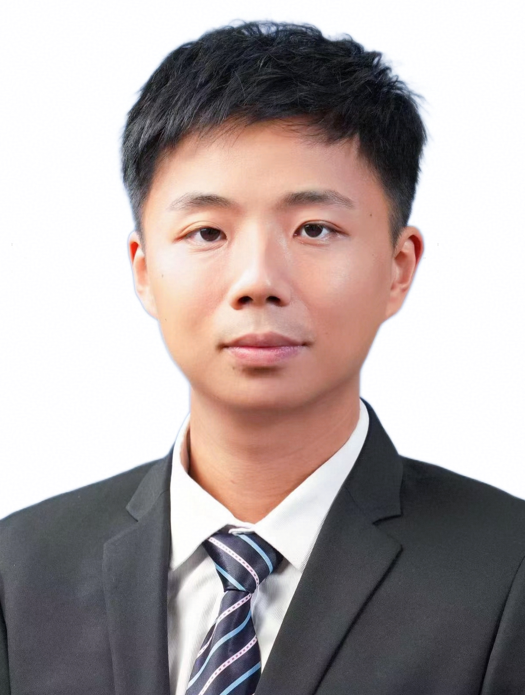}}]{Weidong Mei} (Member, IEEE) received the B.Eng. degree in communication engineering and the M.Eng. degree in communication and information systems from the University of Electronic Science and Technology of China, Chengdu, China, in 2014 and 2017, respectively, and the Ph.D. degree from the NUS Graduate School, National University of Singapore, in 2021 under the Integrative Sciences and Engineering Programme (ISEP) Scholarship. 

He was a Research Fellow with the Department of Electrical and Computer Engineering, National University of Singapore, from July 2021 to January 2023. He is currently a Professor with the University of Electronic Science and Technology of China. His research interests include reconfigurable antennas, near-field beam-engineering, intelligent reflecting surface, wireless drone communications, and convex optimization techniques. 

Dr. Mei has been listed in World's Top 2\% Scientists by Stanford University since 2021. He was the recipients of the Best Paper Award from the IEEE International Conference on Communications in 2021, the Best Student Paper Award from the International Conference on Future Communications and Networks in 2025, and the Outstanding Master's Thesis Award from the Chinese Institute of Electronics in 2017. He mentored his students to win the IEEE ComSoc SPCC Technical Committee Student Challenge and Video Contest Award in 2024. He was honored as the Exemplary Editor of the \textsc{IEEE Open Journal of the Communications Society} in 2024, the Exemplary Reviewer of the \textsc{IEEE Transactions on Communications} in 2019 and 2020, the \textsc{IEEE Open Journal of the Communications Society} in 2021 and 2024, the \textsc{IEEE Wireless Communications Letters} in 2019, 2021, 2022, and 2023, the \textsc{IEEE Communications Letters} in 2021 and 2022, and the {\it SCIENCE CHINA Information Sciences} in 2024 and 2025. He serves as an Associate Editor for the \textsc{IEEE Transactions on Communications}, \textsc{IEEE Internet of Things Journal}, \textsc{IEEE Wireless Communications Letters}, and \textsc{IEEE Open Journal of the Communications Society}, as a Young Editorial Board Member for {\it Journal of Communications and Information Networks} and {\it Acta Electronica Sinica}, and as a Co-Chair of the Workshop on {\it Intelligent Movable and Reconfigurable Antennas for Future Wireless Communication and Sensing} in IEEE Globecom from 2024 to 2026 and IEEE ICC from 2025 to 2026.
\end{IEEEbiography}

\begin{IEEEbiography}[{\includegraphics[width=1in,clip]{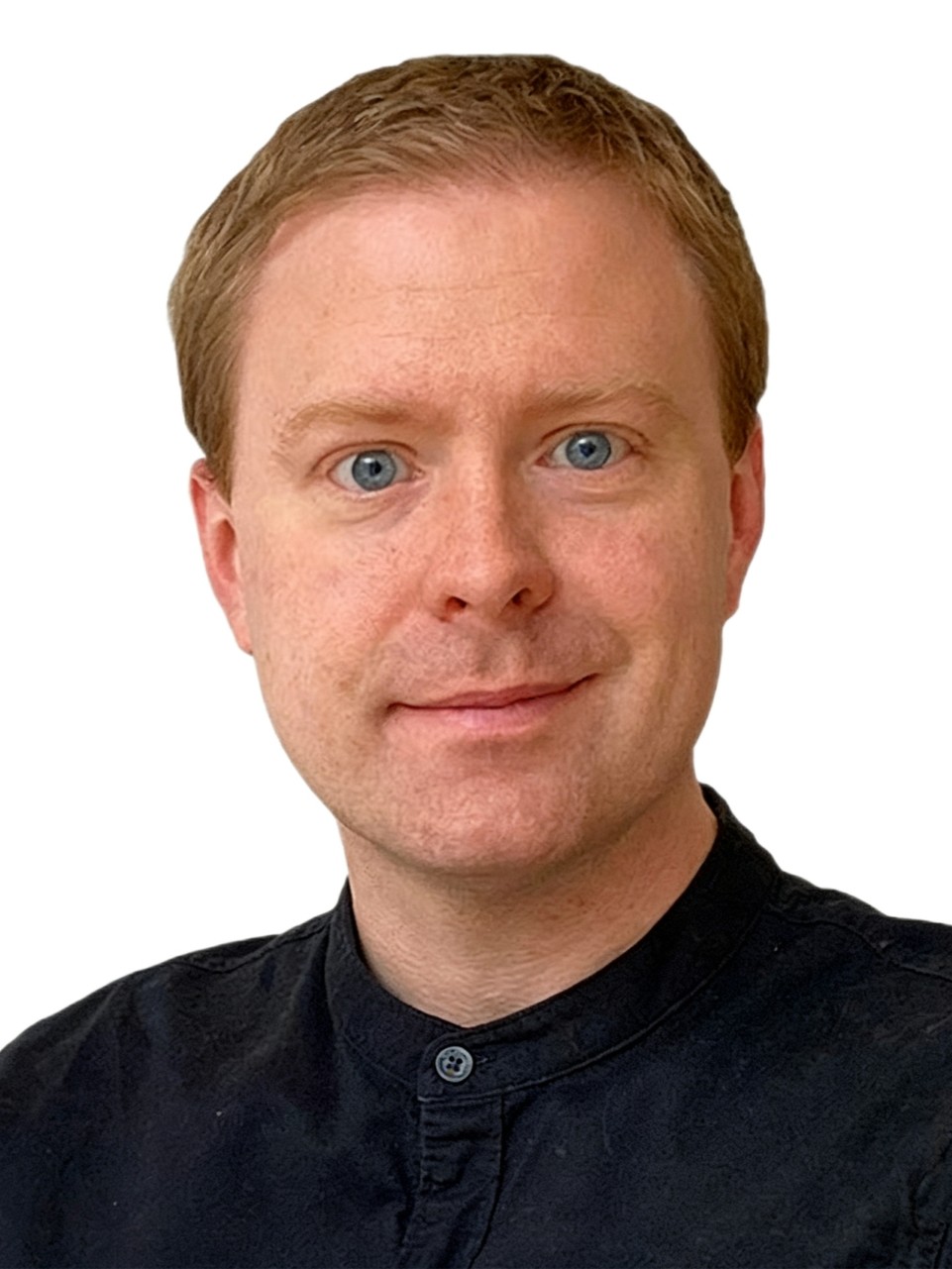}}]{Emil Bj\"ornson}
(Fellow, IEEE) received the M.S. degree in engineering mathematics from Lund University, Sweden, in 2007, and the Ph.D. degree in telecommunications from the KTH Royal Institute of Technology, Sweden, in 2011.

From 2012 to 2014, he was a Post-Doctoral Researcher with the Alcatel-Lucent Chair on Flexible Radio, SUPELEC, France. From 2014 to 2021, he was a faculty member at Link\"oping University, Sweden. He has been a Full Professor of Wireless Communication at KTH since 2020 and currently leads the Department of Communication Systems. He has authored the textbooks \emph{Optimal Resource Allocation in Coordinated Multi-Cell Systems} (2013), \emph{Massive MIMO Networks: Spectral, Energy, and Hardware Efficiency} (2017), \emph{Foundations of User-Centric Cell-Free Massive MIMO} (2021), and \emph{Introduction to Multiple Antenna Communications and Reconfigurable Surfaces} (2024). His research is focused on multi-antenna communications, sensing, and radio resource management, using methods from communication theory, signal processing, and machine learning. He has performed MIMO research since 2006. His papers have received more than 45\,000 citations. He has filed more than 30 patent applications. He is recognized as a Clarivate Highly Cited Researcher. He co-hosts the podcast Wireless Future and has a popular YouTube channel with the same name.

Prof. Bj\"ornson is a Wallenberg Academy Fellow, a Digital Futures Fellow, and an SSF Future Research Leader. He has received the 2018 and 2022 IEEE Marconi Prize Paper Awards in Wireless Communications, the 2019 EURASIP Early Career Award, the 2019 IEEE ComSoc Fred W. Ellersick Prize, the 2019 IEEE Signal Processing Magazine Best Column Award, the 2020 Pierre-Simon Laplace Early Career Technical Achievement Award, the 2020 CTTC Early Achievement Award, the 2021 IEEE ComSoc RCC Early Achievement Award, the 2023 IEEE ComSoc Outstanding Paper Award, the 2024 IEEE ComSoc Stephen O. Rice Prize, and the 2026 EURASIP Athanasios Papoulis Award.
\end{IEEEbiography}

\begin{IEEEbiography}[{\includegraphics[width=1in,clip]{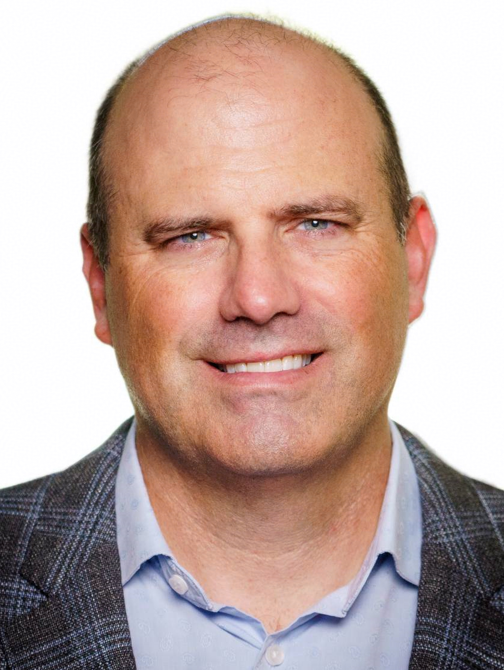}}]{Robert W. Heath Jr.}
(Fellow, IEEE) received the B.S.\ and M.S.\ degrees from the University of Virginia and the Ph.D.\ degree from Stanford University, all in electrical engineering. He is the Charles Lee Powell Chair in Wireless Communications in the Department of Electrical and Computer Engineering at the University of California San Diego, where he directs the Wireless Systems Innovation Laboratory, and he is President and CEO of MIMO Wireless Inc. He was previously the Lampe Distinguished Professor at North Carolina State University and the Cockrell Family Regents Chair in Engineering at The University of Texas at Austin. He authored \emph{Introduction to Wireless Digital Communication} (Prentice Hall, 2017) and co-authored \emph{Millimeter Wave Wireless Communications} (Prentice Hall, 2014) and \emph{Foundations of MIMO Communication} (Cambridge University Press, 2018).

Dr.\ Heath was elected to the National Academy of Engineering in 2025 and is a Fellow of the IEEE, the AAAS, and the National Academy of Inventors. He received the 2025 IEEE/RSE James Clerk Maxwell Medal, the 2021 IEEE Vehicular Technology Society James Evans Avant Garde Award, the 2019 IEEE Kiyo Tomiyasu Award, and the 2017 EURASIP Technical Achievement Award, and his papers have received the 2017 Marconi Prize Paper Award, the 2019 IEEE Communications Society Stephen O.\ Rice Prize, and the 2020 and 2024 IEEE Signal Processing Society Donald G.\ Fink Overview Paper Awards. He was Editor-in-Chief of \emph{IEEE Signal Processing Magazine} from 2018 to 2020. He is a licensed Amateur Radio Operator, a registered Professional Engineer in Texas, and a Private Pilot.
\end{IEEEbiography}

\end{document}